\documentclass{article}
\usepackage{csquotes}
\usepackage[english]{babel}

\usepackage[letterpaper,top=2cm,bottom=2cm,left=3cm,right=3cm,marginparwidth=1.75cm]{geometry}
\usepackage[
  backend=biber,
  style=alphabetic,   % or numeric, authoryear
  sorting=nyt
]{biblatex}

\usepackage{mdframed}
\usepackage{amsmath}
\usepackage{amssymb}
\usepackage{amsmath,amssymb}
\usepackage[most]{tcolorbox}
\usepackage{thmtools}
\usepackage{amsthm}
\usepackage{bbm}
\usepackage{graphicx}
\usepackage{braket}
\usepackage{comment}
\usepackage{xcolor}
\usepackage[colorlinks=true, allcolors=blue]{hyperref}
\theoremstyle{definition} % makes the text upright instead of italic
\newtheorem{theorem}{Theorem}[section] % numbering by section
\newtheorem*{theorem*}{Theorem}

\newtheorem{conjecture}[theorem]{Conjecture}
\newtheorem*{corollary*}{Corollary}
\newtheorem{lemma}[theorem]{Lemma} % numbering by section

\newtheorem{corollary}[theorem]{Corollary} % numbering by section

\newtheorem{proposition}[theorem]{Proposition} % numbering by section
\usepackage{caption}
\usepackage{enumitem}
\usepackage{thm-restate}
\usepackage{multirow}

\newcommand{\Tr}{\operatorname{Tr}}

\newcommand{\tr}{\mathrm{tr}}

\title{Quantum colorings of spheres}
\author{Olivier Lalonde \\
\small Institute for Quantum Computing \\
\small University of Waterloo \\
\small \texttt{olalonde@uwaterloo.ca}
}
\date{}

\begin{document}
\maketitle

\begin{abstract}
    Cameron, Montanaro, Newman, Severini and Winter gave a construction which shows that, for $n \in \{2,4,8\}$, any graph $G$ which admits a real $n$-dimensional orthogonal representation is quantumly $n$-colorable. This result can be recast as the statement that the real sphere $S^{n-1}$ is quantumly $n$-colorable for these values of $n$. We investigate possible extensions of their construction. We first show that their hypothesis that the orthogonal representation be real-valued is required by proving that there is no analogue of this for the complex spheres, which all have quantum chromatic number strictly bigger than the dimension except in two dimensions. We also provide candidate finitary witnesses of this and show for the first time that the real and complex orthogonal ranks are distinct as a byproduct. For the real case, we show that if $S^{n-1}$ is quantumly $n$-colorable, then either $n=2$ or $n$ is a multiple of 4, and show that the converse holds whenever a Hadamard matrix of order $n$ exists. Hence, assuming the Hadamard conjecture, this completely classifies the dimensions to which the CMNSW construction can be extended. Our method of proof involves showing the equivalence between the existence of such a construction and the existence of a maximal code space for Clifford-algebraic errors given a clean ancilla, and we believe that the representation-theoretic techniques we use for tackling the latter problem could be of independent interest. It also follows from this equivalence that $S^{n-1}$ admits a rank-one quantum $n$-coloring if and only if $n \in \{2,4,8\}$, thereby settling a conjecture of Zeng and Zhang, as does the fact that for all $m \geq 1$, there exists a catalytic zero-error remote state preparation protocol for real $m$-qubit states with $m$ bits of communication and which consumes $m$ ebits.
\end{abstract}

\section{Introduction}
\subsection{Background}
In an $n$-coloring game \cite{galliard2002}, two players, traditionally called Alice and Bob, are each given vertices of a publicly known graph $G$ which are promised to either be identical or adjacent, and requested to produce colors in $[n]$ which are identical if and only if their inputs were. In a given model, of interest is the smallest value of $n$ such that this is achievable with probability one for any possible input pair. A strategy achieving this will be called a coloring of the graph. Classically, it is simple to see that such strategies are in one-to-one correspondence with classical colorings of the graph, so that this smallest value of $n$ is simply the chromatic number $\chi(G)$. In the presence of entanglement, this smallest value of $n$ is called the quantum chromatic number of $G$, denoted $\chi_q(G)$ \cite{cameron2006quantumchromaticnumbergraph}. Its study can be traced back to the works \cite{Brassard_1999, avis}, which, in modern notation, showed that, letting $\Omega_{n}$ be the graph with vertex set $\{1,-1\}^{n}$ and with two vertices being adjacent if the corresponding vectors are orthogonal, we have that for all $m$, $\chi_q(\Omega_{4m}) \leq 4m$ (in fact, equality holds, as shown by \cite{Man_inska_2016}). On the other hand, classically, it was shown by \cite{FranklRodl1987ForbiddenIntersections} that $\chi(\Omega_{4m}) \geq \alpha^{4m}$ for some $\alpha > 1$, so that these graphs asymptotically separate the quantum and classical chromatic numbers. This was sharpened by \cite{godsil2005eigenvalueboundsindependentsets} to show that the smallest such graph which exhibits the desired separation is $\Omega_{12}$. 

Though several families of graphs with differing quantum and classical chromatic numbers are now known, these all rely on a handful of known constructions of entangled strategies for the corresponding coloring game, with the most commonly used ones relying on the notion of an orthogonal representation. Following \cite{Lovasz}, given a graph $G$, call an assignment of unit vectors $\{\psi_v\}_{v \in V(G)}$ in $\mathbb{F}^n$ to the vertices of $G$ an orthogonal representation if $\braket{\psi_u | \psi_v} = 0$ whenever $(u,v) \in E(G)$. For example, the aforementioned graphs $\Omega_n$ all have an orthogonal representation in $\mathbb{R}^{n}$, by definition. This notion can be seen to be a linear-algebraic generalization of that of a classical coloring, which is the special case that is obtained when one restricts the $\psi_v$ to be standard basis vectors. We will write $\xi(G)$ and $\xi_\mathbb{R}(G)$ for the least dimension $n$ such that $G$ has an orthogonal representation in $\mathbb{C}^n$ and $\mathbb{R}^n$, respectively. 

We have the following construction due to \cite{cameron2006quantumchromaticnumbergraph}, which is based on the existence of real division algebras in the given dimensions:
\begin{theorem}\label{thm:quaternions}
Given $n \in \{2,4,8\}$, any graph which has an orthogonal representation in $\mathbb{R}^n$ is quantumly $n$-colorable.
\end{theorem}
The most notable example of a quantum-classical separation based on this construction is a graph $G_{14}$ on 14 vertices that was introduced by \cite{mancinska2018odditiesquantumcolorings} and which satisfies $\chi_q(G_{14}) = 4$ and $\chi(G_{14}) = 5$. This graph was shown to be the smallest possible separation between the quantum and classical chromatic numbers by \cite{lalonde2023quantumchromaticnumberssmall}, and its small size has enabled its use in experiments aimed at benchmarking quantum hardware \cite{furches, than}. 

Given its success, it is natural to look for possible extensions of the construction \ref{thm:quaternions}, which could perhaps be used to find other interesting examples of graphs exhibiting separations between their quantum and classical chromatic numbers. It was already noted in \cite{cameron2006quantumchromaticnumbergraph} that their method of proof for theorem \ref{thm:quaternions} strictly requires the orthogonal representation to be real and $n \in \{2,4,8\}$, since $\mathbb{C}^n$ cannot be made into a complex division algebra for any $n > 1$ and since these are exactly the dimensions in which real division algebras exist by Hurwitz's theorem \cite{Hurwitz1898Composition}. Also,  \cite{mancinska2018odditiesquantumcolorings} exhibited a graph $G_{13}$ (from which the aforementioned graph $G_{14}$ is obtained by adding an apex) with an orthogonal representation in $\mathbb{R}^3$ but with $\chi_q(G_{13}) = 4$, thereby showing that no such construction exists for $n=3$.

It is useful to recast construction \ref{thm:quaternions} as follows. For $\mathbb{F} \in \{\mathbb{R}, \mathbb{C}\}$, make the unit sphere in $\mathbb{F}^n$, written $S_\mathbb{F}^{n-1}$, into a graph by taking two points to be adjacent if they are orthogonal. In line with standard notation, we will write $S^{n-1}$ to mean $S_\mathbb{R}^{n-1}$. Note that $S_\mathbb{F}^{n-1}$ contains an $n$-clique given by the standard basis, so that $\chi_q(S_\mathbb{F}^{n-1}) \geq n$. First of all, $S_\mathbb{F}^{n-1}$ trivially has an orthogonal representation in $\mathbb{F}^n$, given by the identity map, so that any theorem to the effect that the quantum $n$-colorability of a graph follows from the existence of an orthogonal representation in $\mathbb{F}^n$ implies that $\chi_q(S_\mathbb{F}^{n-1}) = n$. Conversely, given a graph $G$, the existence of an orthogonal representation of $G$ in $\mathbb{F}^n$ can be seen to be equivalent to the existence of a homomorphism from $G$ into $S_\mathbb{F}^{n-1}$. Since the quantum chromatic number is monotonic under homomorphisms, provided that $\chi_q(S_\mathbb{F}^{n-1}) = n$, the existence of such a homomorphism would imply that $G$ is quantumly $n$-colorable. Hence, the construction \ref{thm:quaternions} can be rephrased as showing that $\chi_q(S^{n-1}) = n$ for $n \in \{2,4,8\}$, and the purpose of this paper is to investigate if this also holds for other values of $n$ or for the complex sphere instead of the real sphere. Note that $S_\mathbb{F}^1$ is a bipartite graph, so that this automatically holds for $n=2$. Hence, we will only be interested in values of $n \geq 3$. 

\subsection{The case of complex spheres and finite witnesses}
Section \ref{sec:complex} investigates possible extensions of the construction \ref{thm:quaternions} to complex orthogonal representations. Our main result is negative. Namely, in subsection \ref{sub:infinitarycomplex}, it is shown:
\begin{theorem} \label{thm:noconstructioncomplex}
For all $n \geq 3$, it holds that $\chi_q(S^{n-1}_\mathbb{C}) > n$. 
\end{theorem}

Note that this is a statement about the quantum chromatic number of an uncountably infinite graph, which is not the preferred object of study in quantum graph theory. In particular, one would like to prove the existence of a family of subgraphs $\{G_n\}_{n \geq 3}$ of the graphs $S^{n-1}_\mathbb{C}$ which witness theorem \ref{thm:noconstructioncomplex}, i.e. which are such that $\xi(G_n) \leq n$ but $\chi_q(G_n) > n$. For the classical chromatic number, this would follow 
from the de Bruijn-Erd\H{o}s theorem:
\begin{theorem}\cite{deBruijnErdos1951}
    If a graph $G$ is such that $\chi(G) < \infty$, then there exists $S \subseteq V(G)$ with $|S| < \infty$ and $\chi(G[S]) = \chi(G)$. 
\end{theorem}
Declare a monotonic graph parameter $f$ to have the de Bruijn-Erd\H{o}s property if the analogous statement holds for it. It is of interest to know which quantum graph parameters have this property. This is investigated in section \ref{sec:debruijn}, where the following is shown. Note that the $d$-dimensional quantum chromatic number and the commuting chromatic number of a graph $G$ correspond to the least number of colors $n$ such that the $n$-coloring game on $G$ has a perfect entangled strategy respectively involving a joint entangled state of local dimension $d$ and an infinite-dimensional entangled state. Rank-$r$ quantum colorings are a proxy for the $d$-dimensional quantum chromatic numbers which are introduced in subsection \ref{subsec:nonlocal}.
\begin{theorem} \label{thm:debruijnerdoswatereddown}
The $d$-dimensional quantum chromatic numbers and the commuting chromatic number all have the de Bruijn-Erd\H{o}s property. Also, for $r, n \geq 1$, if a given graph $G$ is such that all of its finite induced subgraphs have a rank-$r$ quantum $n$-coloring, then so does $G$. On the other hand, the usual quantum chromatic number $\chi_q$ does not.
\end{theorem}
Hence, theorems \ref{thm:noconstructioncomplex} and \ref{thm:debruijnerdoswatereddown} are insufficient to imply the existence of a family of graphs with the aforementioned property, though the implication does hold if a uniform bound is imposed on the local dimension of the shared entangled state in a quantum strategy or if one accepts the conjecture that the quantum and commuting chromatic numbers coincide for all spheres. In subsection \ref{sec:subfinitary}, we describe a family which we conjecture to have this property, and for which we can prove the nonexistence of a restricted form of quantum $n$-coloring. 

\begin{figure}[!t]
    \centering
    \includegraphics[width=0.6\textwidth]{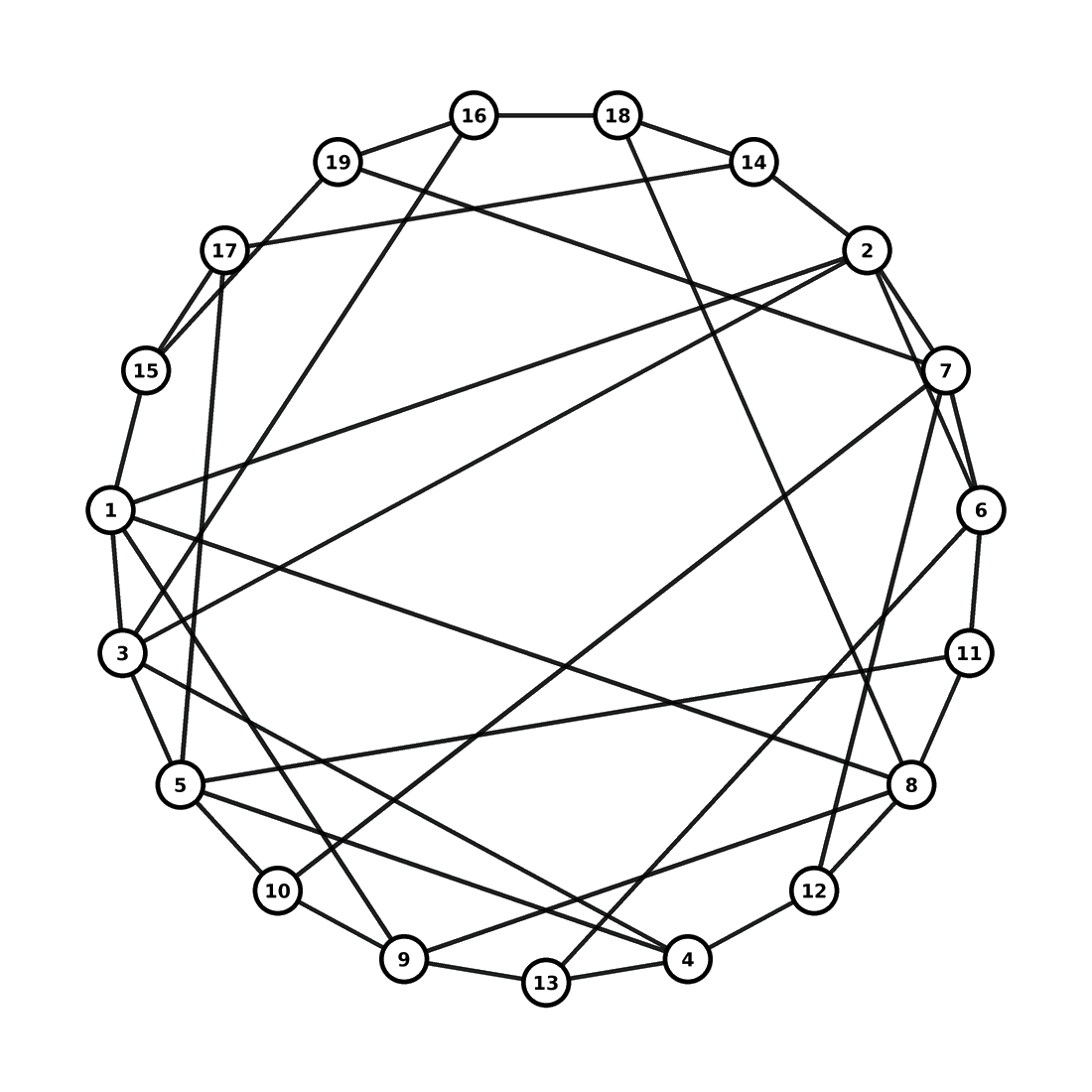}
    \caption{The graph $G_{19}$. The graph6 representation of this graph is \texttt{RxLAKA@AgYAWDGO?O?@??A?W@@OC@\_} .}
    \label{fig:g19_1}
\end{figure}
Let $G_{19}$ be the graph given in figure \ref{fig:g19_1}, and, for $n \geq 3$, set $G_n = G_{19} \vee K_{n-3}$, where $\vee$ denotes the join of two graphs and where $K_{n-3}$ denotes the complete graph on $n-3$ vertices. The graph $G_{19}$ is constructed so that $\xi(G_{19}) = 3$, so that $G_n$ can indeed be realized as a subgraph of $S^{n-1}_\mathbb{C}$ for every $n$. On the other hand, the induced subgraph on the first 13 vertices of $G_{19}$ is precisely the graph $G_{13}$ that was considered in \cite{mancinska2018odditiesquantumcolorings}, which showed that $\chi_q(G_{13}) = 4$, so that $\chi(G_{19}) \geq \chi_q(G_{19}) \geq 4$, and it can be verified that these inequalities are all saturated. One of the main results of \cite{mancinska2018odditiesquantumcolorings} was that $G_{13}$ has the curious property that $\chi^{(1)}_q(G_{13}) = \chi^{(1)}_q(G_{13} \vee K_1) = 4$. We can prove that this does not happen for $G_{19}$: 
\begin{theorem} \label{thm:boundGn}
For all $n$, we have that $\chi^{(1)}_q(G_n) = n+1$.  
\end{theorem}
We conjecture that the same holds for the full quantum chromatic number, though we are unable to prove it. Also, as a byproduct of our analysis, we can prove the following apparently new result (which also follows from corollary \ref{cor:isacore}): 
\begin{theorem} \label{thm:complexrealdistinct}
The graph $G_{19}$ satisfies $\xi(G_{19}) = 3$ and $\xi_\mathbb{R}(G_{19}) = 4$. In particular, the parameters $\xi$ and $\xi_\mathbb{R}$ are distinct.
\end{theorem}
It is an interesting open problem to determine the smallest graph exhibiting a separation between the real and complex orthogonal ranks, as was done for the quantum and classical chromatic numbers by \cite{lalonde2023quantumchromaticnumberssmall}.

\subsection{The case of real spheres}
The purpose of section \ref{sec:realspheres} is to determine for which dimensions construction \ref{thm:quaternions} can be generalized, keeping the assumption that the orthogonal representation be real, which, as was seen in the last subsection, is needed. The conclusion is partly positive and can be summarized by the following theorem. Here $d_{m} = 2^{\lfloor m/2 \rfloor}$ is the dimension of an irreducible representation of the Clifford algebra on $m$ generators, as described in subsection \ref{subsec:reptheory}.
\begin{theorem} \label{thm:realspheres}
For $n \geq 3$, we have: 
\begin{enumerate}
\item If $n$ is not a multiple of $4$, then $\chi_q(S^{n-1}) > n$.
\item If there exists a Hadamard matrix of order $n$, then $\chi^{(d_{n-1})}_q(S^{n-1}) = n$.
\item If $n \geq 2$ is a power of two, then, setting $r = \max(1, \frac{d_{n-1}}{n})$, we have $\chi^{(r)}_q(S^{n-1}) = n$.
\item Given $r \geq 1$, if $rn$ is not divisible by $d_{n-1}$, then $\chi^{(r)}_q(S^{n-1}) > n$. 
\end{enumerate}
\end{theorem}
In particular, it follows from the results of \cite{KharaghaniTayfehRezaie2005Hadamard428} that for all $n \leq 664$, it holds that $\chi_q(S^{n-1}) = n$ if and only if $n=2$ or $n$ is a multiple of 4, thus greatly generalizing the construction \ref{thm:quaternions}. The same holds without the restriction on $n$ under the Hadamard conjecture. Interestingly, the proof of theorem \ref{thm:realspheres} is representation-theoretic and is based on the following equivalences:
\begin{theorem} \label{thm:theproof}
The following are equivalent, for $n \geq 3$, $r \geq 1$:
\begin{enumerate}
    \item It holds that $\chi^{(r)}_q(S^{n-1}) = n$.
    \item There exist unitaries $V_1, \dots, V_n$ on $\mathbb{C}^n \otimes \mathbb{C}^r$ such that, for all $a \neq b$, 
\[(V_a^\dagger V_b)^{\Gamma_1} = -V_a^\dagger V_b\]
    where $\Gamma_1$ denotes the partial transpose on the first register. 
    \item There exists a unitary representation $\mathcal{U}$ of the real Clifford algebra on $n-1$ generators in $U(nr)$ along with a dimension-$r$ Knill-Laflamme code space for $\mathcal{U}$.
\end{enumerate}
\end{theorem}
Subsection \ref{subsec:introitus} shows the equivalence between the first two points, subsection \ref{subsec:equivqecc} shows the equivalence between the last two, and subsection \ref{subsec:theproof} derives theorem \ref{thm:realspheres} from the equivalence between the first and the last. Hence, in effect, theorem \ref{thm:realspheres} is a statement about the existence of high-dimensional code spaces for Clifford-algebraic errors given a clean ancilla, which could be of independent interest.

Theorem \ref{thm:realspheres} has a number of implications which we now discuss.

\subsubsection{A new largest known separation between the quantum and classical chromatic numbers}
As was discussed previously, it has been known for some time that, for some graphs, the classical chromatic number can be exponentially larger than the quantum chromatic number. It is of interest to determine how large the gap between the two can be for a given value of the quantum chromatic number. At the moment, whether this gap can be unboundedly large for a fixed value of the quantum chromatic number is an open problem (see \cite{ciardo2025quantumchromaticgap}), and the largest-known such gap is provided by the graphs $\Omega_{4m}$ that were described earlier, which all have $\chi_q(\Omega_{4m}) = 4m$ but have chromatic number exponential in $m$. It is known that $\chi(\Omega_{4m}) \leq (\beta + o(1))^{4m}$ for $\beta = 2^{1-h(1/4)} \approx 1.1398$, and it is conjectured that this is tight, with this having been shown for certain special families of $m$ (\cite{Frankl1986OrthogonalVectors}, \cite{Ihringer_2019}). Hence, we have:
\[\chi(\Omega_{4m}) \leq (\beta + o(1))^{\chi_q(\Omega_{4m})}\]

In combination with known results, theorem \ref{thm:realspheres} turns out to yield a slightly larger separation between the quantum and classical chromatic numbers than that which is afforded by the graphs $\Omega_{4m}$.  Combining point two of theorem \ref{thm:realspheres}, the Paley construction (see theorem \ref{thm:existencehadamard}) and the asymptotics of \cite{Kumchev2002ConsecutivePrimesAP} for the distribution of gaps between primes congruent to three mod four, we get (noting that a quantum $k$-coloring of $S^{n-1}$ yields a quantum $k$-coloring of $S^{m-1}$ in the obvious way whenever $m \leq n$):
\begin{corollary}
    It holds that $\chi_q(S^{n-1}) = n + O(n^{0.53})$.
\end{corollary}
Let $\Lambda_n$ be the orthogonality graph of the vectors in $\{-1,0,1\}^n \backslash \{0\}$. The graph $\Lambda_3$, for example, is homomorphically equivalent to the graph $G_{13}$ of \cite{mancinska2018odditiesquantumcolorings}. These graphs are not known to have low-dimensional flat orthogonal representations, so that point one of theorem \ref{thm:quantchrom} does not imply the existence of quantum coloring with few colors for the $\Lambda_n$. However, by definition, $\xi_{\mathbb{R}}(\Lambda_n) \leq n$ (and equality holds because the standard basis vectors form an $n$-clique), so that the previous result implies that $\chi_q(\Lambda_n) = n + O(n^{0.53})$. A lower bound of \cite{Rai99} on the classical chromatic numbers of the $\Lambda_n$ then yields:
\begin{corollary} \label{cor:asymptotics}
    It holds that $\chi(\Lambda_n) \geq (\gamma - o(1))^n$, for $\gamma = 2/\sqrt{3} \approx 1.1547$. Hence, we have
    \[\chi(\Lambda_n) \geq (\gamma - o(1))^{\chi_q(\Lambda_n)}\]
\end{corollary}
Note also that, unlike the case of the $\Omega_n$, where $n$ had to be restricted to be a multiple of 4 to get a quantum-classical separation, no such restriction needs to be imposed on the $\Lambda_n$, so that the resulting separation is arguably cleaner. 

\subsubsection{Comparison between the orthogonal rank and the quantum chromatic numbers}
As was mentioned previously, it was shown by \cite{mancinska2018odditiesquantumcolorings} that it is not the case that $\chi_q(G) \leq \xi_\mathbb{R}(G)$ for all graphs by exhibiting a graph $G_{13}$ on 13 vertices with $\xi_\mathbb{R}(G_{13}) = 3$ but $\chi_q(G_{13}) = 4$. Until this work, it was consistent with previous results that, for example, $\chi_q(G) \leq \xi_\mathbb{R}(G)$ does hold for all graphs $G$ such that $\xi_\mathbb{R}(G) \geq 4$. Point one of theorem \ref{thm:realspheres} shows that this is generically false unless $\xi_\mathbb{R}(G)$ is a multiple of 4, thereby greatly generalizing the counterexample of \cite{mancinska2018odditiesquantumcolorings}. 

Similarly, it follows that the observation that one can have $\chi_q(G) = \chi_q(G \vee K_1)$ (which is satisfied by $G_{13}$) holds true for graphs with arbitrarily large orthogonal rank. Indeed, suppose that $n \geq 4$ is such that a Hadamard matrix of order $n$ exists. It follows from either the Sylvester or the Paley constructions that infinitely many such values of $n$ exist. We then have that $\chi_q(S^{n-1}) = n$, by point two of the theorem, and thus $\chi_q(S^{n-2} \vee K_1) = n$ necessarily, since this graph has clique number $n$ and can be realized as a subgraph of $S^{n-1}$. Since $\chi_q(S^{n-2}) \leq \chi_q(S^{n-2} \vee K_1)$ and $\chi_q(S^{n-2}) > n-1$, the latter being given by point one of the theorem, it follows that $\chi_q(S^{n-2}) =n$, as desired.

\subsubsection{The connection with the exact remote state preparation of real states}
Theorem \ref{thm:realspheres} also has implications for the existence of communication-efficient exact remote state preparation schemes for real-valued states. Remote state preparation (\cite{Lo_2000}, \cite{Bennett_2001}) is a relaxed version of the well-known task of quantum teleportation (\cite{PhysRevLett.70.1895}). This task involves two parties, Alice and Bob, who are permitted to share an entangled state of their choosing. Given some publicly known set $T \subseteq S^{n-1}_\mathbb{C}$ and a register size $k$, Alice will be given a classical description of a state $\psi \in T$, and she is to send $c \in [k]$ to Bob so that, at the end of the protocol, Bob has a copy of $\psi$ in his laboratory. For an exact protocol, we require this copy to be perfect. Quantum teleportation is an example of such a protocol with $k = n^2$, whereas extensions of Holevo's theorem imply that any exact protocol must satisfy $k \geq n$ if $T$ contains an orthonormal basis of $\mathbb{C}^n$, as the protocol can then be used to communicate a classical register of size $n$.  

This is the case when $T = S^{n-1}$. One would like to know for which values of $n$ there exists an exact protocol with the lowest communication allowed by information theory, namely $k=n$. This was first shown to be the case for $n=2$ by \cite{Pati_2000}, whose results were then extended to $n \in \{4,8\}$ by \cite{Zeng_2002}. That these dimensions coincide with the ones in construction \ref{thm:quaternions} is not a coincidence: in fact, the two sets of constructions turn out to be entirely equivalent and appear to be independent discoveries. 

In fact, the connection between remote state preparation and quantum colorings runs deeper. On the one hand, the protocols of \cite{Pati_2000} and \cite{Zeng_2002} have the property that the shared entangled state is maximally entangled and that, given her input $\psi \in T$, the measurement that Alice performs on her share of the state to determine her communication $c$, write it $\{P^\psi_c\}_{c \in [k]}$, is projective. Call a protocol of this form restricted. It is then simple to see that the correctness of the protocol implies $P^\psi_c P^\phi_c = 0$ whenever $\braket{\psi | \phi} = 0$. Hence, the existence of a restricted remote state preparation scheme for $T$ with communicated register size $k$ implies that, letting $G_T$ be the graph whose vertices consist of elements of $T$ and with adjacency relations given by orthogonality, the graph $G_T$ is quantumly $k$-colorable, as per theorem \ref{thm:structure2cameron}.

Though the converse is naturally false in general (i.e. the quantum $k$-colorability of $G_T$ need not imply the existence of such a scheme), specifically when $T = S^{n-1}$ and $k=n$, the converse does hold. Figure 2 describes how the equivalence between the first two points of theorem \ref{thm:theproof} can be harnessed to yield a restricted exact remote state preparation scheme for $S^{n-1}$ with communicated register size $n$ whenever it holds that $\chi_q(S^{n-1}) = n$. This protocol recovers the protocols of \cite{Pati_2000} and \cite{Zeng_2002} as the special cases where $n \in \{2,4,8\}$, for which one may take $r=1$ in the statement of the protocol. It was conjectured by \cite{Zeng_2002} that these are the only values of $n$ for which such a protocol with $r=1$ exists. This follows at once from point four of theorem \ref{thm:realspheres}, and is given an alternative proof in section \ref{sec:realspheres}.

\begin{figure}[] \label{fig:theprotocol}
\begin{tcolorbox}[colback=white,colframe=black,sharp corners,boxrule=0.5pt] 
\textbf{Given.}

\begin{itemize}
    \item Natural numbers $n \geq 2$, $r \geq 1$, and unitaries $V_1, \dots,V_n$ on $\mathbb{C}^n \otimes \mathbb{C}^r$ as in point two of theorem \ref{thm:theproof}.
\end{itemize}

\medskip

\noindent\textbf{Protocol.}
Alice and Bob share the standard maximally entangled state $\ket{\Phi}_{A_1A_2, B_1B_2}$ with $A_1$ and $B_1$ of dimension $n$ and $A_2$ and $B_2$ of dimension $r$. Alice is given $\psi \in S^{n-1}$.

\begin{enumerate}
    \item Define the matrices $\{P_c\}_{c \in [n]}$ by:
    \[
        P_c = V_c ((\ket{\psi} \bra{\psi})_{A_1} \otimes I_{A_2}) V_c^\dagger 
    \]
    Lemma \ref{lem:polarization} shows that these matrices specify a projective measurement. Alice measures her share of the state according to this measurement, obtaining outcome $c$. She performs $V_c^\dagger$ on her share of the state and sends $c$ to Bob.
    \item Bob performs $V_c^T$ on his share of the state. 
\end{enumerate}
The final state of the protocol is:
\[\ket{\psi}_{A_1} \otimes \ket{\psi}_{B_1} \otimes \ket{\Phi}_{A_2, B_2}\]
Indeed, corollary \ref{cor:protocolok} shows that the analysis can proceed as if $V_c = I$, after which this follows from the factorization \ref{eq:factorization} and the transpose trick (lemma \ref{lem:transposetrick}). 

\end{tcolorbox}
\caption{The exact remote state preparation protocol proving corollary \ref{cor:rsp}.}
\end{figure}

\begin{corollary} \label{cor:rankone}
    It holds that $\chi^{(1)}_q(S^{n-1}) = n$ if and only if $n \in \{2,4,8\}$.
\end{corollary}

In summary, we record the following implication of corollary \ref{cor:rankone} and point three of theorem \ref{thm:realspheres}, where we specialize to the case where $n$ is a power of two so the relevant resources can be described using information-theoretic terminology.
\begin{corollary} \label{cor:rsp}
    For all $m \geq 1$, there exists an exact remote state preparation protocol for states on $m$ qubits with real-valued amplitudes which uses $m$ bits of communication. This protocol requires $N = \max(m, 2^{m-1} - 1)$ perfect EPR pairs to run but only consumes $m$ of them (i.e. Alice and Bob are left with $N-m$ perfect EPR pairs at the end of the protocol).
\end{corollary}

Remarkably, in addition to being exact, this protocol slightly outperforms previously known approximate protocols for remote state preparation in terms of communication cost and entanglement consumption, though these can handle arbitrary complex states while ours cannot. In fact, it was shown by \cite{Leung_2003} that any exact protocol for remote state preparing arbitrary $m$-qubit states which have certain properties must communicate $2m$ bits. Our protocol does satisfy these properties, so that corollary \ref{cor:rsp} shows that the main result of \cite{Leung_2003} can be circumvented by restricting to real-valued states. Our protocol also saturates known bounds, as it is known that $\approx m$ bits of communication and $\approx m$ EPR pairs are necessary to perform an efficient remote state preparation protocol, even allowing for error, as shown, for example, in \cite{kundu2026nearoptimalentanglementcommunicationtradeoffsremote}. Note that in this last reference, only the remote state preparation of complex states is considered, but the results extend to real-valued states without too much effort. Also, the catalytic nature of the protocol is interesting, as it requires an enormous amount of entanglement but only actually consumes a small fraction of it. We believe that this is the first example of such a protocol in the context of remote state preparation.

\subsubsection{Implications for the rank-$r$ chromatic numbers and the Hadamard conjecture}
Our results also have implications for the study of the rank-$r$ quantum chromatic numbers. It was asked in \cite{cameron2006quantumchromaticnumbergraph} whether it holds that $\chi_q(G) = \chi^{(1)}_q(G)$ for all graphs; this was first explicitly disproved in \cite{lalonde2023quantumchromaticnumberssmall}, which gave the smallest-known counterexample of this in the form of a graph $G_{21}$ on 21 vertices with $\chi_q({G_{21}}) = \chi^{(2)}_q({G_{21}}) = 4$ but $\chi^{(1)}_q(G_{21}) = \chi(G_{21}) = 5$. Though it was already known from complexity-theoretic considerations that, for all $r$, there is a graph $G$ with $\chi_q(G) < \chi^{(r)}_q(G)$, few concrete examples of such graphs are known, and all have been engineered specifically to yield a separation of this form. Theorem \ref{thm:realspheres} shows that the aforementioned conjecture is disproved by a natural class of graphs, namely, real spheres: indeed, it is known that there exists a strictly increasing sequence $\{n_k\}$ of orders for which a Hadamard matrix exists, so that $\chi_q(S^{n_k-1}) = n_k$, but which are such that the least $r_k$ such that $\chi^{(r_k)}_q(S^{n_k-1}) = n_k$ grows exponentially with $n_k$. This is reminiscent of the exponential lower bound on the amount of entanglement required for playing XOR games optimally due to \cite{Slofstra_2011}, and in fact, the underlying mechanisms of both entanglement lower bounds are the same, namely, the representation theory of Clifford algebras. The generalized de Bruijn-Erd\H{o}s theorem implies that finite subgraphs of these spheres exist which exhibit the same features. 

Theorem \ref{thm:realspheres} also has implications for the comparison between the rank-$r$ chromatic numbers. It is simply stated as 'clear' in \cite{cameron2006quantumchromaticnumbergraph} that it holds that $\chi_q^{(r)}(G) \geq \chi_q^{(r')}(G)$ for all graphs $G$ whenever $r \leq r'$. Although this claim was flagged as unproven and possibly false in \cite{lalonde2023quantumchromaticnumberssmall}, until now, whether this really held remained open. This claim is refuted by theorem \ref{thm:realspheres}: for example, points three and four of the theorem imply that $\chi^{(r)}_q(S^{15}) = 16$ if and only if $r$ is a multiple of 8. Again, theorem \ref{thm:debruijnerdoswatereddown} implies the existence of a finite graph $G$ which refutes the claim: indeed, it follows that there exists a subgraph $G'$ of $S^{15}$ with no rank-$9$ 16-coloring. Then, setting $G = G' \sqcup K_{16}$, we have $\chi_q^{(9)}(G) \geq 17$ but $\chi_q^{(8)}(G) = 16$.

Our last point regards the Hadamard conjecture itself. Although it is widely believed that it is true, were it to be false for a given order, as far as we know, the only known proof methods which could establish its falsity would essentially involve exhaustive search, which, even for the smallest order for which the conjecture is not presently known to hold (namely, 668), would almost certainly be computationally completely unfeasible. Interestingly, theorem \ref{thm:realspheres} yields an alternative refutation method: indeed, by point two, it is a prediction of the Hadamard conjecture that any orthogonality graph of a finite subset of $S^{n-1}$ is quantumly $n$-colorable for $n$ a multiple of 4. Hence, were one to exhibit such a graph and prove that it is not quantumly $n$-colorable using, for example, sum-of-squares techniques, the nonexistence of a Hadamard matrix of order $n$ would be established. 

\section{Preliminaries and miscellaneous facts}
We now present definitions and collect general-purpose facts that will be used throughout the rest of the paper.

\subsection{Quantum information and linear algebra} \label{subsec:qm}
We will write $U(n)$ to mean the group of $n \times n$ unitary matrices. We will write $\mathbb{F}_2$ to mean the field with two elements. Given a matrix $M$ mapping $\mathbb{C}^n \otimes \mathbb{C}^{q}$ into $\mathbb{C}^n \otimes \mathbb{C}^{p}$, the partial transpose of $M$ on the first register, denoted $M^{\Gamma_1}$, is defined as follows: writing
\[M = \sum_{a \in [p],b \in [q]} M_{a,b} \otimes \ket{a} \bra{b}\]
We set:
\[M^{\Gamma_1} = \sum_{a \in [p],b \in [q]} M^T_{a,b} \otimes \ket{a} \bra{b}\]
\subsubsection{Maximally entangled states} Given two systems $A$ and $B$ of dimension $d$, the standard maximally entangled state on $AB$, denoted $\ket{\Phi}_{A,B}$, is given by:
\[\ket{\Phi}_{A,B} = \frac{1}{\sqrt{d}} \sum_{k=1}^d \ket{k}_A \ket{k}_B\]
Given systems $A_1$ and $B_1$ of dimension $d_1$ and systems $A_2$ and $B_2$ of dimension $d_2$, the maximally entangled state $\ket{\Phi}_{A_1A_2,B_1B_2}$ is the above state with $A$ taken to be $A_1A_2$ and $B$ taken to be $B_1B_2$. Reordering the registers, we have the factorization:
\begin{equation}  \label{eq:factorization}
    \ket{\Phi}_{A_1A_2,B_1B_2} = \ket{\Phi}_{A_1,B_1} \otimes \ket{\Phi}_{A_2, B_2}
\end{equation}
We have the transpose trick:
\begin{lemma} \label{lem:transposetrick}
    Given a $d \times d$ matrix $M$, we have
    \[(M \otimes I) \ket{\Phi} = (I \otimes M^T) \ket{\Phi}\]
\end{lemma}
A simple corollary is the correctness of our remote state preparation protocol:
\begin{corollary} \label{cor:protocolok}
    Given a $d \times d$ matrix $M$ and $U \in U(d)$, 
    \[(U^\dagger \otimes U^T)(UMU^\dagger \otimes I) \ket{\Phi} = (M \otimes I) \ket{\Phi}\]
\end{corollary}
\begin{proof}
    We have:
    \begin{align*}
    (U^\dagger \otimes U^T)(UMU^\dagger \otimes I) \ket{\Phi} &= (MU^\dagger \otimes I)((I \otimes U^T) \ket{\Phi})\\
    &= (MU^\dagger \otimes I)((U \otimes I) \ket{\Phi})\\
    &= (M \otimes I) \ket{\Phi}
    \end{align*}
    As desired. 
\end{proof}

\subsubsection{Conditions for the vanishing of real and complex quadratic forms}
We have the following standard lemma (see, for example, \cite{axler2015linear}, chapter 7):
\begin{lemma} \label{lem:polarization_1}
Let $M$ be an $n \times n$ matrix. Then:
\begin{enumerate}
    \item It holds that $\braket{\psi | M | \psi} = 0$ for all $\psi \in S^{n-1}_\mathbb{C}$ if and only if $M = 0$.
    \item It holds that $\braket{\psi | M | \psi} = 0$ for all $\psi \in S^{n-1}$ if and only if $M = -M^T$.
\end{enumerate}
\end{lemma}

This result can be extended as follows: 
\begin{lemma} \label{lem:polarization}
For $p, q, n \geq 1$, let $M$ be a matrix encoding a linear map from $\mathbb{C}^n \otimes \mathbb{C}^q$ into $\mathbb{C}^n \otimes \mathbb{C}^p$. Then:
\begin{enumerate}
    \item It holds that $(\bra{\psi} \otimes I_p )M(\ket{\psi} \otimes I_q) = 0$ for all $\psi \in S^{n-1}_\mathbb{C}$ if and only if $M = 0$.
    \item It holds that $(\bra{\psi} \otimes I_p )M(\ket{\psi} \otimes I_q) = 0$ for all $\psi \in S^{n-1}$ if and only if $M = -M^{\Gamma_1}$. 
\end{enumerate}
\end{lemma}
\begin{proof}
    Note that, writing 
    \[M = \sum_{a \in [p],b \in [q]} M_{a,b} \otimes \ket{a} \bra{b}\]
    the statement that $(\bra{\psi} \otimes I_p )M(\ket{\psi} \otimes I_q) = 0$ is equivalent to the statement that $\braket{\psi | M_{a,b} | \psi} = 0$ for all $a,b$. The result then follows from lemma \ref{lem:polarization_1}.
\end{proof}

\subsubsection{Maximal code spaces}
An operator $P$ acting on a Hilbert space $\mathcal{H}$ is said to be a projector if it holds that $P^2 = P$. It is said to be orthogonal if $P^\dagger = P$. In this paper, the term 'projector' will always mean 'orthogonal projector' unless specified otherwise. Given a collection $\mathcal{U}$ of unitaries on $\mathcal{H}$, define the unnormalized noise channel $\Phi_\mathcal{U}$ by:
\[\Phi_\mathcal{U}(X) = \sum_{U \in \mathcal{U}} U X U^\dagger\]
We show the following. The terminology is inspired by quantum error correction, as $P$ is then a projector satisfying the Knill-Laflamme conditions \cite{Knill_2000} of as large a rank as is allowed by dimensional constraints. 
\begin{proposition} \label{prop:codespace}
    For $r, n \geq 1$ let $\mathcal{U} = \{U_1, \dots, U_{n-1}\} \subseteq U(rn)$ and let $P$ be an orthogonal projector on $\mathbb{C}^{rn}$. The following are equivalent:
    \begin{enumerate}
    \item The rank of $P$ is $r$, and it holds that $P U_a P = 0$ for all $a$ and $P U_a^\dagger U_b P = 0$ for all $a \neq b$. 
    \item Letting $P_1 = P$ and $P_{a+1} = U_a P U_a^\dagger$ for $a \in [n-1]$, we have
    \[\sum_a P_a = I\]
    \item Setting $X = P - I/n$, we have $\Phi_\mathcal{U}(X) = -X$.
    \end{enumerate}
    A projector $P$ satisfying the above will be called a maximal code space for the unitaries $\mathcal{U}$.
\end{proposition}
\begin{proof}
We first show that (1) implies (2). Assuming (1) holds, it follows that $P_a P_b = 0$ for all $a \neq b$. Write:
\[Q = \sum_a P_a\]
It is simple to check that $Q$ is a projector. Since its rank is equal to the dimension of the space (as seen by taking the trace), we must have $Q=I$.

Next, suppose that (2) holds. Taking traces shows that the rank of $P$ is $r$. Next, take $b \in [n]$. Then:
\[P_b = P_b I P_b = P_b^3 + \sum_{a \neq b} P_b P_a P_b\]
Therefore,
\[\sum_{a \neq b} P_b P_a P_b = 0\]
Since all the terms in the above sum are positive semidefinite (since $P_b$ is self-adjoint), this gives that $P_b P_a P_b = 0$ whenever $a \neq b$, which implies the desired result.  

Finally, that (2) and (3) are equivalent follows algebraically. 
\end{proof}

\subsubsection{Hadamard matrices}
Given $n \geq 2$, a (real) Hadamard matrix of order $n$ is an $n \times n$ matrix $M$ with entries in $\{-1,1\}$ and with all columns pairwise orthogonal. The simplest example of such a matrix is the familiar Hadamard gate with the $1/\sqrt{2}$ normalization dropped. Such a matrix can only exist if $n = 2$ or $n$ is divisible by four. The classical Hadamard conjecture affirms the converse:
\begin{conjecture}[Hadamard] \label{conj:hadamard}
    There exists a Hadamard matrix of order $n$ whenever $n$ is divisible by four. 
\end{conjecture}
This conjecture has been verified for the following families: 
\begin{theorem} \label{thm:existencehadamard}
For $n$ divisible by four, there exists a Hadamard matrix of order $n$ whenever one of the following is satisfied: 
\begin{enumerate}
\item (Sylvester) If $n=n_1 n_2$ for some $n_1,n_2$ such that Hadamard matrices of orders $n_1$ and $n_2$ exist. In particular, if $n$ is a power of two.
\item (The Paley construction, \cite{Paley}) If $n-1$ is prime.
\item (The result of \cite{KharaghaniTayfehRezaie2005Hadamard428}) If $n \leq 664$. 
\end{enumerate}
\end{theorem}

\subsection{Operator algebras and representation theory} \label{subsec:reptheory}
A $*$-algebra $\mathcal{A}$ is an associative $\mathbb{C}$-algebra equipped with a conjugate-linear involution $*$ satisfying $(xy)^* = y^*x^*$ for all $x,y \in \mathcal{A}$. A $*$-algebra is said to be unital if it has a multiplicative identity, denoted 1: in this work, all $*$-algebras are assumed to be unital. The standard topology to consider on $\mathcal{A}$, called the weak-$*$ topology, is that which is induced by pointwise convergence. Given a set $X$, we will write $\mathbb{C}\langle  X, X^*\rangle$ to mean the free unital $*$-algebra generated by $X$, i.e. the set of noncommutative polynomials in the elements of $X$ and their formal adjoints. An element $x \in \mathcal{A}$ is said to be positive if $x = y^* y$ for some $y \in \mathcal{A}$. A linear functional $f: \mathcal{A} \to \mathbb{C}$ is said to be positive if $f(x) \geq 0$ for all positive $x$. It can be checked that a positive linear functional $f$ must be Hermitian, i.e. $f(x^*) = \overline{f(x)}$ for all $x$. A positive linear functional $f$ is said to be a state if $f(1) = 1$, and a state $f$ is said to be tracial if $f(xy) = f(yx)$ for all $x,y \in \mathcal{A}$. The set of states and tracial states on $\mathcal{A}$ will be denoted by $S(\mathcal{A})$ and $T(\mathcal{A})$ respectively. It can be seen that $T(\mathcal{A})$ is a weak-$*$-closed subset of $S(\mathcal{A})$.  An element $p \in \mathcal{A}$ is said to be a projector if $p^* = p^2 = p$. We will call $\mathcal{A}$ projective if it has a generating set consisting of projectors. Note that a projector $p$ is positive and is such that $1-p$ is also a projector.

An important class of $*$-algebras is given by the vector space of endomorphisms of a vector space $V$, denoted End$(V)$, with multiplication given by composition and with the $*$ operator given by taking the adjoint. We will be assuming that $V$ is finite-dimensional, in which case one may identify End$(V)$ with $M_d(\mathbb{C})$ for $d = \dim V$, the $*$-algebra of $d \times d$ complex matrices. Given $T \subseteq M_d(\mathbb{C})$, we will write $\mathbb{C}\langle T\rangle$ to mean the $*$-subalgebra of $M_d(\mathbb{C})$ generated by the elements of $T$. A $*$-subalgebra $\mathcal{A}$ of $M_d(\mathbb{C})$ is said to be unital if it contains $I_d$ and is said to be irreducible if no nonzero, proper subspace of $\mathbb{C}^d$ is $\mathcal{A}$-invariant. We have the following alternative characterization of irreducibility (see, for example, \cite{LomonosovRosenthal2004Burnside}):
\begin{theorem}[Burnside's theorem] \label{thm:burnside}
A unital $*$-subalgebra of $M_d(\mathbb{C})$ is reducible if and only if it is a proper $*$-subalgebra.  
\end{theorem}

The following standard generalization of Maschke's theorem from the representation theory of groups holds (see, for example, \cite{Putnam2019}):
\begin{theorem}\label{thm:structure}
    Let $\mathcal{A}$ be a $*$-subalgebra of $M_d(\mathbb{C})$. There exists an orthogonal decomposition of $\mathbb{C}^d$ into nonzero $\mathcal{A}$-invariant subspaces
    \[\mathbb{C}^d = \bigoplus_{k=1}^N V_k\]
    such that, for every $k$, the restriction of $\mathcal{A}$ to End($V_k$) is irreducible.
\end{theorem}

Given $m \geq 1$, a collection of unitaries $\mathcal{U} = \{U_1, \dots, U_m\} \subseteq M_d(\mathbb{C})$ is said to form a unitary representation of the real Clifford algebra on $m$ generators if, for all $a, b \in [m]$, one has
\[\{U_a, U_b\} = 2 \delta_{a,b} I\]
where $\{A, B\} = AB + BA$ is the anticommutator of matrices. This representation is said to be irreducible if $\mathbb{C}\langle \mathcal{U} \rangle$ is an irreducible $*$-subalgebra of $M_d(\mathbb{C})$, or equivalently, if it is the whole of $M_d(\mathbb{C})$, by theorem \ref{thm:burnside}. The unitary irreps of the Clifford algebras are classified as follows (see, for example, \cite{Lawson1989}, theorem 5.6): 
\begin{theorem}[Classification of representations of Clifford algebras] \label{thm:clifford}
Suppose the unitaries $U_1, \dots, U_m$ form a unitary representation of the real Clifford algebra in $M_d(\mathbb{C})$. Then, this representation is irreducible if and only if $d = d_m$, where $d_m = 2^{\lfloor m/2 \rfloor}$. Also, given two representations $U_1, \dots, U_m$ and $V_1, \dots, V_m$ in $M_{d_m}(\mathbb{C})$, there exists a unitary $U \in U(d_m)$ and a sign $b = \pm 1$ such that, for all $a$, $U_a = b U V_a U^\dagger$. If $m$ is even, one may take $b=1$. 
\end{theorem}
The most well-known case of the above theorem is when $m=3$, in which case such an irreducible representation is given by the three non-identity Pauli matrices, which underpin quantum teleportation (\cite{PhysRevLett.70.1895}): 
\[
X = \begin{pmatrix} 0 & 1 \\ 1 & 0 \end{pmatrix}, \quad
Y = \begin{pmatrix} 0 & -i \\ i & 0 \end{pmatrix}, \quad
Z = \begin{pmatrix} 1 & 0 \\ 0 & -1 \end{pmatrix}
\]
For larger values of $m$, an irreducible representation of the Clifford algebra on $m$ generators can be constructed using the Pauli matrices by means of the Jordan-Wigner transformation (\cite{jordan1928pauli}). Given such an irreducible representation, the following corollary, which follows from theorems \ref{thm:structure}, \ref{thm:clifford}, lets one parameterize all representations:
\begin{corollary} \label{cor:superclifford}
There exists a unitary representation of the real Clifford algebra in $M_d(\mathbb{C})$ if and only if $d$ is divisible by $d_m$. For such a value of $d$, the representations are exactly the ones of the following form. Pick an irreducible representation $\Gamma_1,\dots,\Gamma_m$ of the real Clifford algebra in $M_{d_m}(\mathbb{C})$. Given natural numbers $p, q$ with $(p+q)d_m = d$ and a unitary $U \in U(d)$, take the corresponding representation to be:
\[U_a = U((I_p \otimes \Gamma_a) \oplus (I_q \otimes (-\Gamma_a)))U^\dagger\]
If $m$ is even, one may take $q=0$.
\end{corollary}

Given a unitary representation $U_1,\dots,U_m$ of the real Clifford algebra, for $S \subseteq [m]$, define the Clifford monomial $U_S$ by
\[U_S = \prod_{a \in S} U_a\]
The degree of a monomial $U_S$ is defined to be $|S|$. We will be writing $\Lambda_d$ to denote the vector space generated by the degree-$d$ monomials. The following is a simple consequence of the fact that the Clifford generators anticommute and square to the identity: 
\begin{lemma} \label{lem:productmonomials}
    Given Clifford monomials $U_S$, $U_T$, we have that $U_S U_T \in \Lambda_{|S \Delta T|}$, where $\Delta$ denotes the symmetric difference of two sets. 
\end{lemma}

Note that because the $U_a$ are self-adjoint, square to the identity and pairwise anticommute, the $*$-algebra they generate is spanned by the $\Lambda_d$. Hence, if the representation is irreducible, by theorem \ref{thm:burnside}, the $\Lambda_d$ span the entire matrix algebra. We have the following useful formulas:
\begin{lemma} \label{lem:cliffordmonomials}
Let $U_S$ and $U_T$ be two Clifford monomials. Setting $r = |S| |T| - |S \cap T|$, we have
\[U_S U_T = (-1)^r U_T U_S\]
Also, $U_S^\dagger = b U_S$, where $b=1$ if $|S| \equiv 0,1 \mod 4$ and $b=-1$ otherwise. 
\end{lemma}

\subsection{Graph theory background}
This subsection provides background and notation for general graph theory and graph parameters other than the quantum chromatic numbers, which will be handled in the next subsection.

\subsubsection{General graph theory and homomorphisms}
In this paper, all graphs under consideration are undirected and simple, so that a graph $G$ is specified by a vertex set $V(G)$ and an edge set $E(G) \subseteq {V(G) \choose 2}$. The graph $G$ is said to be finite if $|V(G)| < \infty$ and infinite otherwise. In this paper, though we will be allowing for some of our graphs to be infinite, a certain amount of regularity will implicitly be assumed, for example, the existence of a homomorphism $f: G \to H$ for some finite graph $H$. Given a graph $G$ and $S \subseteq V(G)$, the induced subgraph $G[S]$ refers to the graph with vertex set $S$ and edge set $E(G) \cap (S \times S)$. Given two graphs $G$ and $H$, their disjoint union, denoted $G \sqcup H$, is the graph made up of disjoint copies of $G$ and $H$. Their join, denoted $G \vee H$, is the graph obtained by taking the disjoint union of $G$ and $H$ and making every vertex of $G$ adjacent to every vertex of $H$. In this work, $K_k$ will stand for the complete graph on $k$ vertices. 

Given two graphs $G$ and $H$, a graph homomorphism from $G$ to $H$ is a map $f: V(G) \to V(H)$ such that $(f(u), f(v)) \in E(H)$ for all $(u,v) \in E(G)$. In this work, the word 'homomorphism' will always refer to a graph homomorphism unless stated otherwise. We will write $\text{Hom}(G,H)$ to denote the set of homomorphisms from $G$ to $H$.  We write $\text{End}(G)$, the set of endomorphisms of $G$, for $\text{Hom}(G,G)$, and we write $\text{Aut}(G)$, the automorphism group of $G$, to mean the set of elements in $\text{End}(G)$ which are permutations and whose inverses are also endomorphisms. Two homomorphisms $f,g \in \text{Hom}(G,H)$ will be called adjacent if $(f(x), g(x)) \in E(H)$ for all $x \in V(G)$. Two graphs $G$ and $H$ are said to be homomorphically equivalent if $\text{Hom}(G,H), \text{Hom}(H,G) \neq \emptyset$. A graph $G$ is said to be a core if End$(G)$ = Aut$(G)$, and it is said to be bipartite if $\text{Hom}(G,K_2) \neq \emptyset$. We have the following simple properties of homomorphisms, which readily follow from the definitions:
\begin{proposition} \label{prop:homomorphisms}
    Letting $G,H, K$ be graphs, we have:
    \begin{enumerate}
        \item If $f \in \text{Hom}(G,H)$ and $g \in \text{Hom}(H,K)$, then $g \circ f \in \text{Hom}(G,K)$.
        \item If $f, g \in \text{Hom}(G,H)$ are adjacent and $h \in \text{Hom}(H,K)$, then $h \circ f$ and $h \circ g$ are adjacent. 
        \item If $f, g \in \text{Hom}(H,K)$ are adjacent and $h \in \text{Hom}(G,H)$, then $f \circ h$ and $g \circ h$ are adjacent. 
        \item If $G$ is a core and $H$ can be realized as a proper subgraph of $G$, then $\text{Hom}(G,H) = \emptyset$. 
    \end{enumerate}
\end{proposition}

\subsubsection{Spheres and Grassmannians}
For a given field $\mathbb{F} \in \{\mathbb{R}, \mathbb{C}\}$, we will write $S_\mathbb{F}^{n-1}$ to mean the unit sphere in $\mathbb{F}^n$. In line with conventional notation, we will write $S^{n-1}$ to mean $S_{\mathbb{R}}^{n-1}$. The Grassmannian $Gr_{\mathbb{F}}(k,n)$ refers to the set of linear subspaces of dimension $k$ of $\mathbb{F}^n$. We will be writing $Gr(k,n)$ to mean $Gr_{\mathbb{C}}(k,n)$. Elements of this set are in one-to-one correspondence with rank-$k$ orthogonal projectors over $\mathbb{F}^n$, and notation will be abused slightly by identifying them for notational convenience.  These spaces are all made into graphs by taking two points to be adjacent if they are orthogonal, where two subspaces $T$ and $U$ are said to be orthogonal if $\braket{\psi | \phi} = 0$ whenever $\psi \in T,\phi \in U$. 

It is simple to see that $Gr_\mathbb{F}(k, 2k)$ is bipartite for all $k \geq 1$. The spaces $Gr_\mathbb{F}(1,n)$ are known as the real and complex projective spaces $\mathbb{FP}^{n-1}$ and are the quotient of the sphere $S_\mathbb{F}^{n-1}$ by the equivalence relation under which two points are identified if they are collinear. Hence, as was already implicit in \cite{cameron2006quantumchromaticnumbergraph} and was spelled out in \cite{godsil2024quantumindependencechromaticnumbers}, the graphs $\mathbb{FP}^{n-1}$ and $S_\mathbb{F}^{n-1}$ can be seen to be homomorphically equivalent. 

We have the following Uhlhorn-type rigidity result for homomorphisms from spheres into Grassmannians, due to \cite{TianXu2025WignerRankN}, which will be foundational to our analysis in sections \ref{sec:complex} and \ref{sec:realspheres}. Note that, translating their work into graph-theoretic terminology, their result only classified elements  of Hom$(Gr_\mathbb{F}(1,n), Gr(r, rn))$. However, it is simple to see that any $f \in \text{ Hom}(S_\mathbb{F}^{n-1}, Gr(r, rn))$ is such that $f(\psi) = f(\phi)$ whenever $\psi$ and $\phi$ are collinear, so that Hom$(Gr_\mathbb{F}(1,n), Gr(r, rn))$ and Hom$(S^{n-1}_\mathbb{F}, Gr(r, rn))$ are completely equivalent objects and the given statement follows from their result. 
\begin{theorem}[\cite{TianXu2025WignerRankN}, corollary 1.3] \label{thm:tianxu}
For $r \geq 1$, $n \geq 3$, $\mathbb{F} \in \{\mathbb{R}, \mathbb{C}\}$, we have that Hom$(S_\mathbb{F}^{n-1}, Gr(r, rn))$ consists of all maps $f$ of the following form, where $\mathbb{C}^{rn}$ is being identified with $\mathbb{C}^{n} \otimes \mathbb{C}^r$ in the natural way. For a given choice of a unitary $U$ acting on $\mathbb{C}^{n} \otimes \mathbb{C}^r$ and $p, q \geq 0$ with $p+q = r$, set 
\[f(\psi) = M(\psi)M(\psi)^\dagger\]
where
\[M(\psi) = U ((\ket{\psi} \otimes I_p) \oplus (\ket{\overline{\psi}} \otimes I_q)) \]
If $\mathbb{F} = \mathbb{R}$, one may take $q=0$ in the above (as complex conjugation has no effect on real vectors).
\end{theorem}
The special case of the above with $r=1$ was already shown by \cite{pankov2020geometricapproachwignertypetheorems} and implies: 
\begin{corollary} \label{cor:isacore}
    For $n \geq 3$, we have that $\mathbb{FP}^{n-1}$ is a core.
\end{corollary}
The above results do not hold in two dimensions. For example, the projective line $\mathbb{CP}^1$ is bipartite, and is therefore not a core (since $K_2$ can be realized as the proper subgraph corresponding to the two standard basis vectors). 

\subsubsection{Non-quantum graph parameters}
A graph parameter $f$ is a map from (sufficiently regular) graphs into real numbers. We call $f$ monotonic if it holds that $f(G) \leq f(H)$ whenever $\text{Hom}(G,H) \neq \emptyset$: hence, we have $f(G) = f(H)$ whenever $G$ and $H$ are homomorphically equivalent. All graph parameters under consideration in this paper will be monotonic. 

Let $G$ be a graph. For $n \in \mathbb{N}$, a (classical) $n$-coloring of $G$ is an element of $\text{Hom}(G, K_n)$. The graph is said to be $n$-colorable if such an element exists, and the chromatic number $\chi(G)$ is the smallest $n$ such that $G$ is $n$-colorable. An $n$-dimensional orthogonal representation $f$ of $G$ over the field $\mathbb{F}$ is an element of $\text{Hom}(G, S_\mathbb{F}^{n-1})$, and the orthogonal rank $\xi_\mathbb{F}(G)$ is defined to be the least $n$ such that such an element exists. Call $f$ flat if, for all vertices $x$ and all standard basis vectors $k$, we have that $|\braket{k | f(x)}|^2 = 1/n$, and define the flat orthogonal rank $\xi'(G)$ to be the least $n$ such that a flat $n$-dimensional orthogonal representation of $G$ over $\mathbb{C}$ exists.  A $k$-clique of $G$ is a subset $S \subseteq V(G)$ with $|S| = k$ such that $G[S] \cong K_k$, and the clique number of $G$, denoted $\omega(G)$, is the largest $k$ such that $G$ has a $k$-clique.

The following theorem collects basic properties of the previously defined graph parameters:
\begin{theorem} \label{thm:graphparameters}
    We have, for any graph $G$ and any $f \in \{\omega, \xi_\mathbb{R}, \xi_\mathbb{C}, \xi', \chi\}$:
    \begin{enumerate}
        \item The parameter $f$ is monotonic.
        \item We have $\omega(G) \leq \xi_\mathbb{C}(G) \leq \xi_\mathbb{R}(G), \xi'(G) \leq  \chi(G)$.
        \item For all $n \geq 2$, we have that $\text{Hom}( S_\mathbb{C}^{n-1}, S_\mathbb{R}^{2n-1}) \neq \emptyset$. Hence, $\xi_\mathbb{R}(S^{n-1}_\mathbb{C}) \leq 2n$, and so $\xi_\mathbb{R}(G) \leq 2\xi_\mathbb{C}(G)$.
        \item If $f \neq \omega$, we have that $f(G) = 2$ if and only if $\chi(G) = 2$, i.e. if and only if $G$ is bipartite. 
        \item If $f \neq \xi'$, for any graph $H$, we have $f(G \vee H) = f(G) + f(H)$.
        \item For all $n \geq 3$, it holds that $\chi(S^{n-1}) \geq n+1$, with equality holding for $n=3$. Asymptotically, letting $c = 2/\sqrt{3} \approx 1.15, C = \sqrt{2} \approx 1.41$, it holds that $(c - o(1))^n \leq \chi(S^{n-1}) \leq (C+o(1))^n$. In particular, any graph $G$ satisfies $\chi(G) \leq (C+o(1))^{\xi_\mathbb{R}(G)} \leq (C+o(1))^{2\xi_\mathbb{C}(G)}$.
    \end{enumerate}
\end{theorem}
\begin{proof}
The first five points are standard and simple to see. For the first part of the sixth point, we note that it suffices to handle the case of $n=3$ because, for all $k \geq 1$, $S^2 \vee K_k$ can be realized as a subgraph of $S^{2+k}$ and we have that $\chi(S^{2} \vee K_k) = \chi(S^2) + k = 4+k$ by point five. For the $n=3$ case, the first reference to have given the result in this exact form appears to be \cite{godsil2012colouringsphere}, though a stronger result was shown much earlier by \cite{Simmons_1976}. The lower bound also follows from corollary \ref{cor:isacore}: indeed, $K_3$ can be realized as a proper subgraph of $\mathbb{RP}^2$. The latter is a core, so that Hom$(\mathbb{RP}^2, K_3) = \emptyset$. Since $\mathbb{RP}^2$ is homomorphically equivalent to $S^{2}$, it follows that Hom$(S^2, K_3) = \emptyset$. For the second part of the statement, the lower bound is due to \cite{Rai99} (see also \cite{bucic2024geometricgraphsexponentialchromatic}, theorem 5), while the upper bound is due to \cite{PROSANOV20183123}.
\end{proof}

\subsection{Nonlocal games and quantum chromatic numbers} \label{subsec:nonlocal}
\subsubsection{Correlation sets}
Given input sets $X, Y$ and output sets $A, B$, a correlation set $C(X,Y,A,B)$ is specified by a collection of probability distributions $\{p(a, b | x, y)\}_{x \in X, y \in Y, a \in A, b \in B}$. The physical scenario that we have in mind is that Alice and Bob are spatially separated, given $x \in X$ and $y \in Y$ respectively, and requested to provide outputs $a \in A$, $b \in B$ according to the corresponding probability distribution. Contrary to most previous work, we allow $X$ and $Y$ to be infinite sets, but we will always require $A$ and $B$ to be finite sets. The reason for this is that, as it turns out, the existing theory accommodates this extension seamlessly, while allowing the answer sets to be infinite causes a great deal more disruption. 

The set of correlations corresponding to locally-realistic theories, denoted $C_{lhv}$, is the convex hull of the set of deterministic correlations. Quantum theory gives rise to three correlation sets $C_{q} \subsetneq C_{qa} \subsetneq C_{qc}$, corresponding respectively to the set of correlations which can be achieved with finite-dimensional tensor-product strategies, the closure of the former, and infinite-dimensional commuting-operator strategies. The latter admits the following alternative characterization, which is more convenient to work with. Define the $*$-algebra $\mathcal{A}$ to be the free $*$-algebra generated by the $\{e^x_a\}_{x \in X, a \in A}$, $\{f^y_b\}_{y \in Y, b \in B}$ subject to the $e^x_a$ and the $f^y_b$ being orthogonal projectors and satisfying the PVM relations:
\[\sum_{a} e^x_a = \sum_b f^y_b = 1\]
and subject to $[e^x_a, f^y_b] = 0$ for all $x,y,a,b$. We then have:
\begin{proposition} \label{prop:correlations}
    A given correlation $p$ is such that $p \in C_{qc}$ if and only if there is a state $f$ on $\mathcal{A}$ such that $p(a,b | x,y) = f(e^x_a f^y_b)$ for all $x,y,a,b$.
\end{proposition}

A nonlocal game (\cite{cleve2010consequenceslimitsnonlocalstrategies}) corresponds to a linear functional on the set of correlations of a special form, and is specified by a probability distribution $\pi$ on $X \times Y$ and a subset $W \subseteq X \times Y \times A \times B$. The scenario we have in mind is that the players are provided inputs $x$ and $y$ sampled from $\pi$, are requested to produce outputs $a \in A$, $b \in B$ as before, and are declared to have won the game if $(x,y,a,b) \in W$. Of interest in a given model is the highest possible success probability that is achievable by a strategy. In pseudo-telepathy (\cite{Brassard_1999}), one is interested in whether this success probability can be one or not, in which case the probability distribution $\pi$ is only relevant in that it specifies the subset $T \subseteq X \times Y$ of possible input pairs.

\subsubsection{Coloring games and chromatic numbers}
As first defined by \cite{galliard2002}, coloring games are instances of pseudo-telepathy games and were discussed in the introduction. Given $t \in \{lhv, q, qa, qc\}$, we can define the quantum chromatic number of a graph $G$, $\chi_t(G)$, to be the least number of colors $n$ such that the $n$-coloring game on $G$ can be won with probability one in the model. By definition, we have $\chi_{qc}(G) \leq \chi_{qa}(G) \leq \chi_q(G) \leq \chi_{lhv}(G) = \chi(G)$ for all graphs. 

Define the $d$-dimensional quantum chromatic number of a graph $G$, denoted $\chi^{[d]}_q(G)$, to be the least number of colors $n$ such that there exists a perfect entangled strategy for the $n$-coloring game on $G$ involving a shared entangled state of local dimension $d$. The following structure theorem is due to \cite{cameron2006quantumchromaticnumbergraph}, and greatly simplifies the study of the quantum chromatic number:
\begin{theorem} \label{thm:structure2cameron}
There exists a perfect entangled strategy for the $n$-coloring game on a graph $G$ involving a shared entangled state of local dimension $d$ if and only if there exists a collection $\{P^v_a\}_{v \in V(G), a \in [n]}$ of projective measurements over $\mathbb{C}^d$, one for each vertex $v$, such that for all $(v,w) \in E(G)$ and all $a \in [n]$, we have that $P^v_a P^w_a = 0$. The corresponding strategy is as follows: Alice and Bob share the maximally entangled state $\ket{\Phi}_{A,B}$ with $A$ and $B$ of dimension $d$. Given inputs $v$ and $w$, Alice and Bob measure their shares of the state with the projective measurements $\{P^v_a\}_{a \in [n]}$ and $\{\overline{P}^w_b\}_{b \in [n]}$, respectively.
\end{theorem}
A collection of projective measurements as in the statement of the theorem will be called a $d$-dimensional quantum coloring of $G$. The following is immediate:
\begin{proposition}
    For a given graph $G$, it holds that $\chi^{[1]}_q(G) = \chi(G)$. The $\chi^{[d]}_q(G)$ form a non-increasing sequence in $d$, and we have:
    \[\chi_q(G) = \lim_{d \to \infty} \chi^{[d]}_q(G)\]
\end{proposition}

Define a rank-$r$ quantum $n$-coloring of $G$ to be a collection of $n$ pairwise adjacent homomorphisms $f_a: G \to Gr(r, rn)$. The rank-$r$ quantum chromatic number of $G$, written $\chi^{(r)}_q(G)$, is defined to be the least $n$ such that $G$ has a rank-$r$ $n$-coloring. The proof of proposition 1 in \cite{cameron2006quantumchromaticnumbergraph} shows:
\begin{theorem} \label{thm:structurecameron}
For all $r\geq 1$ and all graphs $G$, it holds that $\chi^{(r)}_q(G) \leq \chi^{[r]}_q(G)$.
\end{theorem}
A word should be said about our notation. The $d$-dimensional quantum chromatic numbers $\chi^{[d]}_q$ are non-standard notation, and the rank-$r$ quantum chromatic numbers $\chi^{(r)}_q$, which serve as a proxy for them, are the more conventional object of study in the literature. We will be considering both because, though the rank-$r$ chromatic numbers will be essential to our analysis, they also have a number of unfortunate properties: for example, contrary to what was claimed in \cite{cameron2006quantumchromaticnumbergraph}, it is not true generically that $\chi^{(r+1)}_q(G) \leq\chi^{(r)}_q(G)$ for all $r \geq 1$ and all graphs $G$, nor is it even obviously true that the existence of a rank-$r$ quantum $n$-coloring of a graph implies the existence of a rank-$r$ quantum $(n+1)$-coloring of the same graph. This forces one into some gymnastics, such as in the statement of theorem \ref{thm:debruijnerdoswatereddown}.  Hence, it seems to us that the $\chi^{[d]}_q$ should be considered the more central object, with the parameters $\chi^{(r)}_q$ serving as convenient proxies for them through theorem \ref{thm:structurecameron}. 

The following theorem collects some of the known properties of the quantum chromatic numbers which were not already mentioned:
\begin{theorem} \label{thm:quantchrom}
For all graphs $G$:
\begin{enumerate}
\item It holds that $\omega(G) \leq \chi_{qc}(G) \leq \chi_q(G) \leq \chi^{(1)}_q(G) \leq \xi'(G) \leq \chi(G)$.
\item We have that $\chi_{qc}(G) = 2$ precisely when $\chi(G) = 2$, i.e. precisely when it is bipartite.
\item Given the description of the adjacency matrix of a finite graph $G$, determining whether $\chi_{qc}(G) = 3$ is co-RE-complete and determining whether $\chi_{q}(G) = 3$ is RE-complete. For any $d, n$, determining whether $\chi^{[d]}_{q}(G) = n$ is a decidable problem, so that there exists a sequence $G_d$ of graphs with $\chi_q(G_d) = 3$ and $\chi^{[d]}_{q}(G_d) = 4$ for all $d$.
\end{enumerate}
\end{theorem}
\begin{proof}
For the first point, the first inequality follows from the fact that $\omega(G) \leq \vartheta(\overline{G}) \leq \chi_{qc}(G)$, where $\vartheta$ is Lovász's number, as shown by \cite{paulsen2013quantumchromaticnumbersoperator}. The second and third inequalities are definitional, and the last two were shown in \cite{cameron2006quantumchromaticnumbergraph}. The second point is shown in \cite{paulsen2013quantumchromaticnumbersoperator}. For the third point, the undecidability results were shown by \cite{harris2023universalitygraphhomomorphismgames}, building on earlier work of \cite{Slofstra_2019}, \cite{ji2013}, \cite{ji2022mipre}. The decidability of $\chi^{[d]}_{q}(G) = n$ is due to the fact that this can be recast into the existential theory of the reals, which was shown to be decidable by \cite{Tar51}. 
\end{proof}

\section{Elements of infinite quantum graph theory} \label{sec:debruijn}
This section investigates generalizations of the de Bruijn-Erdős theorem. Letting $f$ be a monotonic graph parameter and letting $G$ be an infinite graph such that $f(G)$ is well-defined, we see that
\begin{equation} \sup_{S \subseteq V(G), |S| < \infty} f(G[S]) \leq f(G) \label{eq:dbe} \end{equation} 
The parameter $f$ will be said to have the de Bruijn-Erdős property if equality always holds in the above. This terminology is motivated by the following classical result:
\begin{theorem} \cite{deBruijnErdos1951}
The chromatic number satisfies the de Bruijn-Erdős property.
\end{theorem}
Since this paper is concerned with the study of uncountably infinite graphs (namely, the real and complex spheres), it is of interest to determine for which other graph parameters this holds. Though the original proof of \cite{deBruijnErdos1951} does not seem to generalize conveniently beyond the classical chromatic number, the proof due to \cite{Gottschalk1951} does without too many hurdles. In fact, it is even possible to generalize the above theorem to general pseudo-telepathy games. Declare a given space of strategies to have the de Bruijn-Erdős property if the following is true: given a pseudo-telepathy game $(X,Y,A,B,W)$ with $|A|, |B| < \infty$ and which does not have a winning strategy in the model, there exist finite subsets $X' \subseteq X$, $Y' \subseteq Y$ such that, given the promise that Alice's and Bob's inputs lie in $X'$ and $Y'$, there does not exist a winning strategy in the model for this restricted game. We show
\begin{theorem} \label{thm:dbepseudotelepathy}
The following settings have the de Bruijn-Erdős property:
\begin{enumerate}
    \item The finite-dimensional tensor product setting, subject to the condition that Alice and Bob must use an entangled state of local dimension $d$, for any constant $d$ (and in particular $d=1$, i.e. the classical setting).
    \item The commuting operators setting.
\end{enumerate}
\end{theorem} 
Theorem \ref{thm:debruijnerdoswatereddown} is simply the special case of the above where coloring games are considered. On the other hand, theorem \ref{thm:dbepseudotelepathy} is false for the general tensor-product setting, even for coloring games, as the following shows:
\begin{proposition}
        The quantum chromatic number $\chi_q$ does not have the de Bruijn-Erdős property.
\end{proposition}
\begin{proof}
    By theorem \ref{thm:quantchrom}, there exists a sequence of graphs $\{G_d\}$, all with $\chi_q(G_d) = 3$ but $\chi^{[d]}_q(G_d) = 4$. It is then easy to see that their disjoint union $G$ is such that every finite subgraph of $G$ has quantum chromatic number at most $3$, but $G$ doesn't, for if it did, we would have that $\chi_q^{[d^*]}(G) \leq 3$ for some $d^*$, which would give that $\chi^{[d^*]}_q(G_d) = 3$ for all $d$, which would be a contradiction. 
\end{proof}
One could declare a given graph parameter to have the countable de Bruijn-Erdős property if equation \ref{eq:dbe} holds with the supremum taken over all countable subgraphs instead. The quantum chromatic number is easily seen to have the countable de Bruijn-Erdős property: indeed, we have shown that a given graph $G$ is such that, for every $d$, there exists a finite subgraph $G_r$ of $G$ with $\chi^{[d]}_q(G_r) = \chi^{[d]}_q(G)$. Letting $G'$ be the union of all the $G_r$, which evidently has countable vertex set, we get that $\chi_q(G) = \chi_q(G')$.

In order to prove theorem \ref{thm:dbepseudotelepathy}, we introduce notation. Let $\Omega$ be a set and let $\mathcal{S}$ be a collection of subsets of $\Omega$ such that for all $S_1,\dots,S_n \in \mathcal{S}$, there exists $T \in \mathcal{S}$ such that $S_i \subseteq T$ for every $i$, and which is such that: 
\[\bigcup_{S \in \mathcal{S}} S = \Omega\]
We take a presheaf on $\mathcal{S}$ to mean an assignment to every $S \in \mathcal{S}$ of a topological space $Z_S$ as well as a collection of continuous inclusion maps $g_{T,S}: Z_T \to Z_S$ for every $T,S \in \mathcal{S}$ with $S \subsetneq T$ (setting $g_{T,T} = \text{id}_{Y_T}$), satisfying $g_{R,S} \circ g_{T,R} = g_{T,S}$ for all $S,R,T \in \mathcal{S}$ with $S \subsetneq R \subsetneq T$. We start with the following lemma, proved here for completeness:
\begin{lemma} \label{lem:hausdorff}
    Let $X,Y,Z$ be topological spaces with $Y$ Hausdorff, and endow $A = X \times Y \times Z$ with the product topology. Given a continuous function $f: X \to Y$, define the set $S \subseteq A$ by:
    \[S = \{a \in A  \mid  f(a_X) = a_Y\}\]
    Then $S$ is closed in the product topology.
\end{lemma}
\begin{proof}
    We prove that $S^c$ is open. Take $a \notin S$. Because $Y$ is Hausdorff, there exist disjoint open sets $U, V \subseteq Y$ with $f(a_X) \in U$, $a_Y \in V$. Let $V'$ and $W$ be the pullbacks of $V$ and $f^{-1}(U)$ in $A$, which are open by definition: then $V' \cap W$ is open and contains $a$. Also, if $b \in S \cap W$, then $f(b_X) \in U$, and so $b \notin V'$ since $U$ and $W$ are assumed disjoint: hence $S$ and $V' \cap W$ and $S$ are disjoint, which completes the proof.
\end{proof}
We can show the following semi-standard result, which is, for example, a slight generalization of a proof method already used in \cite{Gottschalk1951}, and which we will be using as a template for the rest of the subsection.
\begin{theorem} \label{thm:inversesystem}
Let $\Omega$ and $\mathcal{S}$ be as above and take a presheaf on $\mathcal{S}$ such that the corresponding topological spaces $Z_S$ are compact and Hausdorff. Writing
\[ Z = \prod\limits_{S \in \mathcal{S}} Z_S\]
\[Z' = \{z \in Z \mid g_{T,S}(z_T) = z_S, \;\; \forall S,T \in \mathcal{S},S \subsetneq T\}\]
Then, endowing $Z$ with the product topology, $Z'$ is compact, and is nonempty if the $Z_S$ are all nonempty.
\end{theorem} 
\begin{proof}
    For given $S, T \in \mathcal{S}$ with $S \subsetneq T$, write:
    \[C_{T,S} = \{z \in Z \mid g_{T,S}(z_T) = z_S\}\]
    And for given $T \in \mathcal{S}$, write 
    \[C_T =  \bigcap \limits_{S \subsetneq T} C_{T,S}  \]
    Note that
    \[Z' = \bigcap_S C_{S}\]
    The rest of the proof proceeds as follows. Note that because the $Z_S$ are assumed compact, Tychonoff's theorem implies that so is $Z$. Lemma \ref{lem:hausdorff} implies that the $C_{T,S}$ are closed, and therefore so are the $C_S$, and hence so is $Z'$. Since $Z' \subseteq Z$, we get that $Z'$ is compact. \\

    For the second part of the proof, suppose that the $Z_S$ are all non-empty. Note that this implies that so are the $C_S$. We will show that $C_T \subseteq C_S$ whenever $S \subsetneq T$. This implies that the $C_S$ have the finite intersection property: indeed, given $S_1,\dots,S_n \in \mathcal{S}$, it is assumed that there exists $T$ such that $S_1,\dots,S_n \subseteq T$, hence
    \[\bigcap_{i=1}^n C_{S_i} \supseteq C_T \neq 0\]
    Since the $C_S$ are compact, this implies that their intersection is nonempty, as desired. Note first that given $R \subsetneq S \subsetneq T$, we have that $C_{T,S} \cap C_{T,R} \subseteq C_{S,R}$ for all $R \subsetneq S \subsetneq T$: indeed, if $z \in C_{T,S} \cap C_{T,R}$, then $z_R = g_{T,R}(z_T) = g_{S,R}(g_{T,S}(z_T)) = g_{S,R}(z_S)$. Therefore, 
    \begin{align*}
    C_T &\subseteq C_{T,S} \cap \left(\bigcap_{R \subsetneq S} C_{T,R}\right) \\
    &= \bigcap_{R \subsetneq S} (C_{T,S} \cap C_{T,R})\\
    &\subseteq \bigcap_{R \subsetneq S} C_{S,R}\\
    &= C_S
    \end{align*}
    which completes the proof.
\end{proof}

This lets us prove:
\begin{proof}[Proof of point one of theorem \ref{thm:dbepseudotelepathy}]
Let $\mathfrak{G} = (X,Y,A,B,W)$ be such a pseudo-telepathy game. We prove the contrapositive using \ref{thm:inversesystem}. We take the set $\Omega$ to be $X \times Y$ and we take the collection $\mathcal{S}$ to be the set of finite rectangles $S = X' \times Y'$ in $X \times Y$. It can be seen that $\mathcal{S}$ has the required properties. Write $M_d$ for the set of $d \times d$ positive semidefinite matrices with largest eigenvalue at most one, and, given a finite set $R$, define $P_{d,R}$ by:
\[P_{d,R} = \left\{E \in (M_d)^{R} \mid \sum_{r \in R } E_r= I\right\}\]
The set $P_{d,R}$ encodes the set of all POVMs on $\mathbb{C}^d$ indexed by elements of $R$. Since $M_d$ is compact, we see that so is $P_{d,R}$. 

Define a presheaf on $\mathcal{S}$ as follows: given $S=  X' \times Y' \in \mathcal{S}$, set the topological space $Z_S$ to be:
\[Z_S = \{z \in P^{X'}_{d,A} \times P^{Y'}_{d, B} \times M_{d^2} \mid \Tr(z_3) = 1 \text{ and, } \forall (x,y,a,b) \in W^c \cap (X' \times Y' \times A \times B),\; \Tr(z_3 ((z_1)_{x,a}\otimes (z_2)_{y,b})) = 0\} \]
We see that $Z_S$ encodes the space of all winning strategies involving an entangled state of local dimension $d$. We have that $Z_S$ is Hausdorff and compact with the usual identifications and that it is nonempty by assumption. The inclusion maps are defined in the natural way: given $X'' \subseteq X'$, $Y'' \subseteq Y'$, $S' = X'' \times Y''$, the map $g_{S,S'}: Z_S \to Z_{S'}$ is defined by $(g(z)_1)_x = (z_1)_x$, $(g(z)_2)_y = (z_2)_y$, $g(z)_3 = z_3$. These inclusion maps can be seen to satisfy the compatibility requirement. \\

Hence, set
\[Z = \prod_{S \in \mathcal{S}} Z_S\]
and let $z$ be an element of $Z'$ as in the statement of theorem \ref{thm:inversesystem}. Pick any $S \in \mathcal{S}$ and set:
\[\rho = (z_S)_3\]
The fact that $z \in Z'$ shows that this $\rho$ is independent of $S$. Indeed, given $S, S' \in \mathcal{S}$, take $T \in \mathcal{S}$ such that $S,S' \subseteq T$. Then:
\[(z_S)_3 = g_{T,S}((z_T))_3 = (z_T)_3 = g_{T,S'}((z_T))_3 = (z_{S'})_3\]
Given $x \in X, y \in Y$, pick some $S \in \mathcal{S}$ with $(x,y) \in S$ and set, for all $a \in A, b \in B$,
\[E^x_a = ((z_S)_1)_{x,a}\]
\[F^y_b = ((z_S)_2)_{y,b}\]
For any fixed $x$, the $E^x_a$ encode a POVM on $\mathbb{C}^d$, and likewise for the $F^y_b$. For the same reason as before, their definition does not depend on the choice of $S$. Also, the fact that $z_S \in Z_S$ implies that, for all $(x,y,a,b) \notin W$, we have
\[\tr(\rho (E^x_a \otimes F^y_b)) = 0\]
as desired. 
\end{proof}
The following can be shown along very similar lines:
\begin{proposition}
    The orthogonal ranks $\xi_\mathbb{C}, \xi_{\mathbb{R}}, \xi'$ all have the de Bruijn-Erdős property. Also, for given $r, n \geq 1$, if an infinite graph $G$ is such that every finite subgraph of $G$ has a rank-$r$ quantum $n$-coloring, then so does $G$.
\end{proposition}

The rest of the section is devoted to proving the second point of theorem \ref{thm:dbepseudotelepathy}. The following is standard:
\begin{lemma}[Cauchy-Schwarz] \label{fact:cauchyschwarz}
A positive linear functional $f$ is such that for all $x, y \in \mathcal{A}$, 
\[|f(x^* y)|^2 \leq f(x^* x) f(y^* y)\]
\end{lemma}
We can now show: 
\begin{lemma} \label{fact:projectors}
Let $\mathcal{A}$ be a unital $*-$algebra. If $p_1,\dots,p_n \in \mathcal{A}$ are projectors and $f$ is a state, then $|f(p_1\dots p_n)| \leq 1$.
\end{lemma}
\begin{proof}
Using lemma \ref{fact:cauchyschwarz}, we have
\begin{align*}
|f(p_1 \dots p_n)|^2 &= |f((p_1 \dots p_n)^\dagger 1)|^2\\
&\leq f(p_1 p_{2} \dots  p_n p_{n-1} p_1)
\end{align*} 
We prove that this last expression is bounded by one by induction on $n$. This is trivially true if $n=0$ because $f$ is a state. If the result is true for $n-1$, writing $a = p_1 \dots p_{n-1}$, then, for any projector $q$, $aqa^* = (qa^*)^* (qa^*)$ and so $aqa^*$ is positive. Hence, 
\begin{align*}
f(a p_n a^\dagger) &\leq f(a p_n a^\dagger) + f(a(1-p_n)a^\dagger)\\
                &= f(a a^\dagger)\\
                &\leq 1
\end{align*} 
where the last inequality can be seen to follow from the induction hypothesis.
\end{proof}
We can then show:
\begin{proposition}
Let $\mathcal{A}$ be a projective $*$-algebra. It holds that $S(\mathcal{A})$ is compact in the weak-$*$ topology.
\end{proposition}
\begin{proof}
By hypothesis, there exists a generating subset $P \subseteq \mathcal{A}$ consisting of projectors. Let $X$ be the set of all words over $P$, with the empty word being identified as the identity. Note that since projectors are self-adjoint, the linear span of the algebra generated by $P$ is all of $\mathcal{A}$. We define a presheaf over $X$. We take $\mathcal{S}$ to be the collection of all finite subsets of $X$ which contain the identity. Given such a $S = \{w_1,w_2,\dots,w_n\}$, write $V_S$ for the linear subspace of $\mathcal{A}$ spanned by $w_1,\dots,w_n$, and take $Y_S$ to be the set of linear functionals $f: V_S \to \mathbb{C}$ satisfying $|f(w_i)| \leq 1$ for every $i$, $f(1) = 1$ and $f(x) \geq 0$ for every positive $x \in V_S$. It can be seen that $Y_S$ is compact and Hausdorff under the natural topology coming from viewing every element in $Y_S$ as an element of $\mathbb{C}^n$. Given $T$ with $S \subsetneq T$, the inclusion map $g_{T,S}: Y_T \to Y_S$ is defined straightforwardly: given $f \in Y_T$ and given $p \in V(S)$, we set $(g_{T,S}(f))(p) = f(p)$. Clearly, this map is continuous and satisfies the required consistency relations. 

Given this data, let $Y'$ be the set that is promised by theorem \ref{thm:inversesystem}. Define the map $\phi: Y' \to L(\mathcal{A})$ as follows. Given $y \in Y'$, we define $\phi(y)$ as follows: for a given $p \in \mathcal{A}$, take $S \in \mathcal{S}$ such that $p \in V_S$ and set $\phi(y)(p) = y_S(p)$. Note first that this is well-defined: indeed, if $S,S'$ are such that $p \in V_S$, $p \in V_{S'}$, then taking $T \in \mathcal{S}$ such that $S,S' \subseteq T$ we have that $y_S(p) = g_{T,S}(y_T)(p) = y_T(p) = g_{T,S'}(y_T)(p) = y_{S'}(p)$. From the way $Y_S$ was defined, it follows that $\phi(y)$ is always a state, and more generally it can be seen that $\phi: Y' \to S(\mathcal{A})$ is a bijection (with the reverse direction being given by fact \ref{fact:projectors}). It is simple to see that $\phi$ is continuous, and hence $S(\mathcal{A})$ is the image of a compact topological space under a continuous map and is therefore itself compact.
\end{proof}

From here, the proof of point two of theorem \ref{thm:dbepseudotelepathy} follows in a similar way as the first part.

\section{Quantum colorings of complex spheres} \label{sec:complex}
This section is devoted to proving the lower bound on the quantum chromatic numbers of complex spheres that is promised by theorem \ref{thm:noconstructioncomplex}. We first handle the infinitary case, which turns out to only require elementary linear algebra. We then present a candidate for a family of finitary witnesses of this separation in the form of the graphs $G_n = G_{19} \vee K_{n-3}$. 

\subsection{The infinitary case} \label{sub:infinitarycomplex}

To prove that $\chi_q(S_\mathbb{C}^{n-1}) > n$ for all $n \geq 3$, we will prove the following stronger statement, which implies it in view of theorem \ref{thm:structurecameron}:

\begin{proposition} \label{prop:complexcase}
    For all $n \geq 3$, $r \geq 1$, there do not exist three pairwise adjacent elements of Hom($S_\mathbb{C}^{n-1}, Gr(r, rn)$).
\end{proposition}
We note that the above proposition is false if three is replaced by two, provided that $n$ is even: indeed, taking $r=1$, one can choose $f_1(\psi) = \ket{\psi} \bra{\psi}$ and $f_2(\psi) = U \ket{\overline{\psi}} \bra{\overline{\psi}} U^\dagger$ for some antisymmetric unitary $U$, for example $U = Y^{\oplus \frac{n}{2}}$.

We will need the following lemma:
\begin{lemma} \label{lem:complexstructure}
Let $r,n \geq 1$, let $p_1, p_2,q_1,q_2 \geq 0$ be such that $p_1 + q_1 = p_2 + q_2 = r$, and suppose that $U$ is a unitary on $\mathbb{C}^{rn}$ such that, for all $\psi \in S^{n-1}_\mathbb{C}$, 
\[
((\bra{\psi} \otimes I_{p_1}) \oplus (\bra{\overline{\psi}} \otimes I_{q_1})) U ((\ket{\psi} \otimes I_{p_2} ) \oplus (\ket{\overline{\psi}} \otimes I_{q_2})) = 0 \label{eq:eqlemma}
\]
Then we must have $p_1 = q_2$, $q_1 = p_2$, and $U$ must be of the form
\[U = \begin{pmatrix} 0 & B \\ C & 0 \end{pmatrix}\]
With $B \in U(np_1)$, $C \in U(np_2)$.
\end{lemma}
\begin{proof}
    Write $U$ in block form:
    \[U = \begin{pmatrix} A & B \\ C & D \end{pmatrix}\]
    with $A$ of dimension $(np_1) \times (np_2)$, $B$ of dimension $(np_1) \times (nq_2)$, $C$ of dimension $(nq_1) \times (np_2)$ and $D$ of dimension $(nq_1) \times (nq_2)$. Then, for all $\psi$, 
    \[0 = ((\bra{\psi} \otimes I_{p_1}) \oplus (\bra{\overline{\psi}} \otimes I_{q_1})) U ((\ket{\psi} \otimes I_{p_2} ) \oplus (\ket{\overline{\psi}} \otimes I_{q_2})) = \begin{pmatrix} (\bra{\psi} \otimes I_{p_1}) A (\ket{\psi} \otimes I_{p_2}) & (\bra{\psi} \otimes I_{p_1}) B (\ket{\overline{\psi}} \otimes I_{q_2}) \\(\bra{\overline{\psi}} \otimes I_{q_1})  C (\ket{\psi} \otimes I_{p_2}) & (\bra{\overline{\psi}} \otimes I_{q_1}) D (\ket{\overline{\psi}} \otimes I_{q_2})\end{pmatrix}\]
We see that point one of lemma \ref{lem:polarization} implies that $A=0, D=0$, since complex conjugation is a permutation of $S_\mathbb{C}^{n-1}$. Hence, $U$ is of the form:
\[U = \begin{pmatrix} 0 & B \\ C & 0 \end{pmatrix}\]
Because $U$ is a unitary, the rows of $B$ all have unit norm and are pairwise orthogonal, and likewise for the columns of $B$. Since the row and column spaces of a matrix have the same dimension, this forces $B$ to be a square matrix, which gives $p_1 = q_2$ and that $B$ is a unitary. The same reasoning applied to $C$ yields the rest of the lemma. 

\end{proof}
We can now prove: 

\begin{proof}[Proof of proposition \ref{prop:complexcase}]
We proceed by contradiction. Let $f_i$ be three such homomorphisms, and let $U_1, U_2, U_3$ and $(p_1,q_1),(p_2,q_2),(p_3, q_3)$ be the data that is promised by theorem \ref{thm:tianxu}. Note that $P \to U_1^\dagger P U_1$ is an automorphism of $Gr(r, rn)$: hence, composing the $f_i$ with this automorphism, by proposition \ref{prop:homomorphisms}, we can assume without loss of generality that $U_1 = I$. For any given $\psi$, because the $f_a(\psi)$ are pairwise orthogonal subspaces by assumption, it holds that for $a \neq b$,
\begin{equation} \label{eq:myequation}
((\bra{\psi} \otimes I_{p_b} ) \oplus (\bra{\overline{\psi}} \otimes I_{q_b})) U_b^\dagger U_a (( \ket{\psi} \otimes I_{p_a}) \oplus (\ket{\overline{\psi}} \otimes I_{q_a})) = 0
\end{equation}
Specializing equation \ref{eq:myequation} to $b=1$, we have that for $a \in \{2,3\}$ and for all $\psi$, 
\[((\bra{\psi} \otimes I_{p_1} ) \oplus (\bra{\overline{\psi}} \otimes I_{q_1})) U_a (( \ket{\psi} \otimes I_{p_a}) \oplus (\ket{\overline{\psi}} \otimes I_{q_a})) = 0\]
Writing $p_1 = q, q_1 = p$, lemma \ref{lem:complexstructure} implies that $p_a = p$, $q_a = q$ and that we can write:
\[U_2 = \begin{pmatrix} 0 & B_1 \\ C_1 & 0\end{pmatrix}\]
\[U_3 = \begin{pmatrix} 0 & B_2 \\ C_2 & 0\end{pmatrix}\]
with $B_1, B_2 \in U(nq)$ and $C_1, C_2 \in U(np)$. Hence:
\[U_2^\dagger U_3 = \begin{pmatrix} C_1^\dagger C_2 &  0\\ 0 & B_1^\dagger B_2\end{pmatrix}\]
Also, specializing equation \ref{eq:myequation} to $b=2$, $a=3$ and invoking lemma \ref{lem:complexstructure}, we get that $p = q$ and that, for some $B_3,C_3 \in U(np)$, 
\[U_2^\dagger U_3 = \begin{pmatrix} 0 &  B_3\\ C_3 & 0 \end{pmatrix}\]
which is a clear contradiction.
\end{proof}

\subsection{The finitary case: the case of the graph $G_{19}$} \label{sec:subfinitary}
We just showed that the unit sphere in $\mathbb{C}^n$ is never quantumly $n$-colorable for $n \geq 3$, thereby showing that the construction behind theorem \ref{thm:quaternions} fails if the assumption that the orthogonal representation be real is dropped. We would like to have a finitary witness for this, i.e. a finite graph $G_n$ with $\xi(G_n) = n$ and $\chi_q(G_n) > n$.  Here we present a candidate for such a family of graphs. Define the graph $G_{19}$ as the orthogonality graph of the set of vectors $V$ given as follows, where $i$ is the imaginary unit:
\begin{equation}
\begin{aligned}
V := {}&
\left\{
\begin{pmatrix}1\\0\\0\end{pmatrix},
\begin{pmatrix}0\\1\\0\end{pmatrix},
\begin{pmatrix}0\\0\\1\end{pmatrix}
\right\} \\[0.5em]
&{}\cup
\left\{
\begin{pmatrix}1\\1\\0\end{pmatrix},
\begin{pmatrix}1\\-1\\0\end{pmatrix},
\begin{pmatrix}1\\0\\1\end{pmatrix},
\begin{pmatrix}1\\0\\-1\end{pmatrix},
\begin{pmatrix}0\\1\\1\end{pmatrix},
\begin{pmatrix}0\\1\\-1\end{pmatrix}
\right\} \\[0.5em]
&{}\cup
\left\{
\begin{pmatrix}1\\1\\1\end{pmatrix},
\begin{pmatrix}1\\1\\-1\end{pmatrix},
\begin{pmatrix}1\\-1\\1\end{pmatrix},
\begin{pmatrix}1\\-1\\-1\end{pmatrix}
\right\} \\[0.5em]
&{}\cup
\left\{
\begin{pmatrix}1\\0\\-i\end{pmatrix},
\begin{pmatrix}0\\1\\-i\end{pmatrix},
\begin{pmatrix}1\\i\\0\end{pmatrix}
\right\} \\[0.5em]
&{}\cup
\left\{
\begin{pmatrix}1\\1\\i\end{pmatrix},
\begin{pmatrix}1\\-i\\i\end{pmatrix},
\begin{pmatrix}1\\-i\\1\end{pmatrix}
\right\}
\end{aligned}
\end{equation}

We note that the induced subgraph on the first 13 vertices of $G_{19}$ is exactly the graph $G_{13}$ which was considered in \cite{mancinska2018odditiesquantumcolorings}. Hence, we immediately have that $\chi_q(G_{19}) \geq \chi_q(G_{13}) = 4$. We have that equality holds as $G_{19}$ is classically four-colorable, with a possible four-coloring being given by
\[(1, 2, 3, 2, 1, 3, 1, 3, 2, 3, 2, 4, 1, 1, 2, 1, 3, 2, 3)\]

On the other hand, by definition, $\xi(G_{19}) \leq 3$, and equality holds since the first three vertices form a 3-clique. The purpose of this subsection is to prove theorem \ref{thm:boundGn}. Note that the upper bound is immediate as $\chi(G_{n}) = n+1$ for all $n$. Hence, it suffices to show the following:
\begin{theorem} \label{thm:non3coloring}
    For all $n \geq 3$, we have that $\chi^{(1)}_q(G_{n}) > n$. 
\end{theorem} 
We need the following structural input about the graph $G_{19}$, which is proven with a computer using the algorithm described in the last section of \cite{lalonde2023quantumchromaticnumberssmall}. The rest of the proof, however, proceeds in a human-checkable fashion. 
\begin{lemma} \label{lem:notcollinear}
Let $u$ and $v$ be distinct non-adjacent vertices of $G_{19}$ and let $G'$ be the graph on 18 vertices that is obtained by identifying them. Then, $\xi(G') > 3$. 
\end{lemma}
\begin{figure}[!t]
    \centering
    \includegraphics[width=0.6\textwidth]{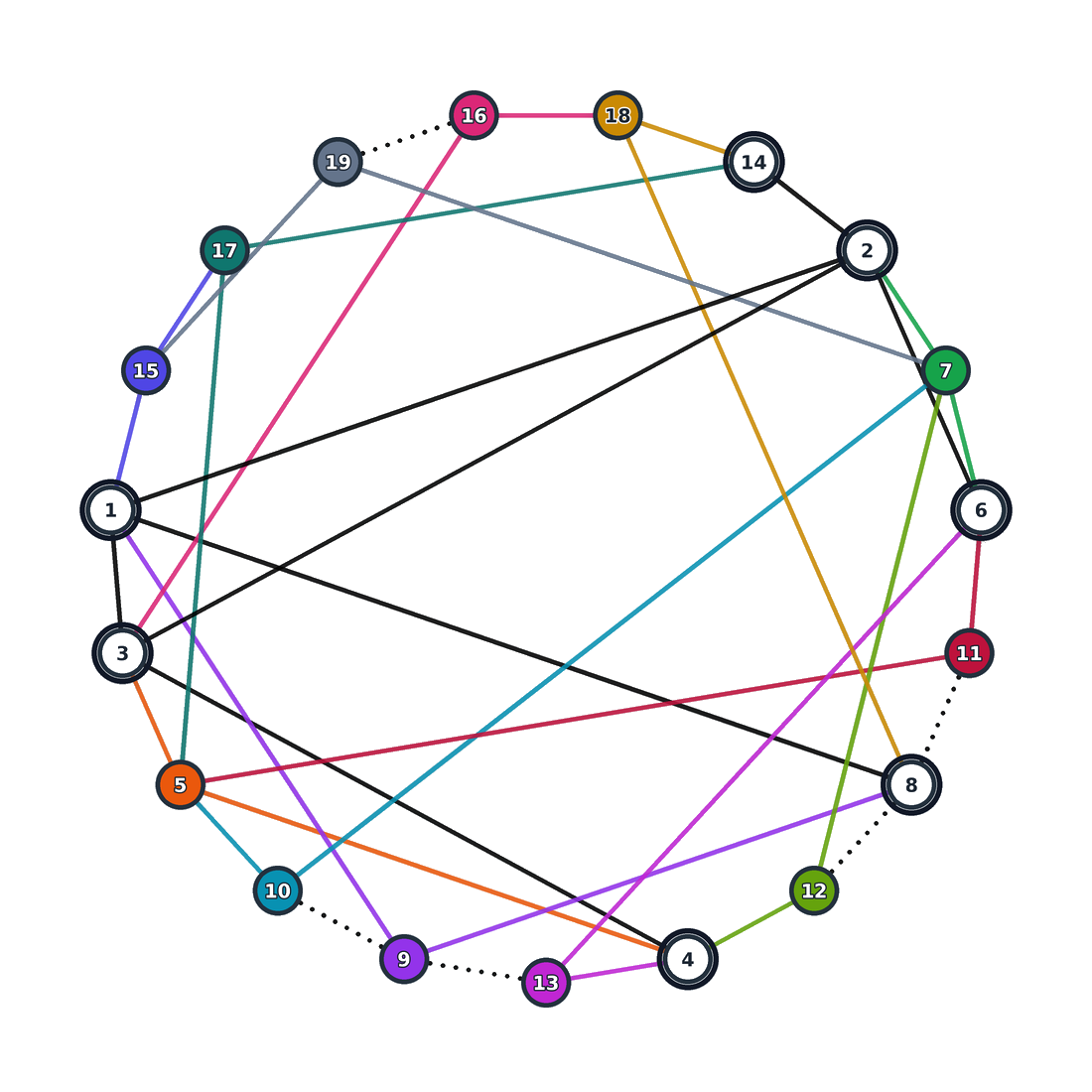}
    \caption{An alternative representation of the graph $G_{19}$, used in the proof of theorem \ref{thm:classificationreps}.}
    \label{fig:g19_2}
\end{figure}

Given this input, we can now analytically classify all possible orthogonal representations of the graph $G_{19}$. This is achieved by manually running an algorithm used in the Kochen-Specker literature for proving that a given graph does not have a real three-dimensional orthogonal representation and which is due to \cite{Uijlen_2014}. As a byproduct, theorem \ref{thm:complexrealdistinct} is proven. 
\begin{theorem} \label{thm:classificationreps}
Up to a global unitary and up to scaling, the orthogonal representations of the graph $G_{19}$ in $\mathbb{C}^3$ are exactly those of the following form, for some $\alpha, \beta \in \mathbb{C}$ with $|\alpha| = |\beta| = 1$ and for some $b \in \{1, -1\}$, which will be referred to as being in normal form:
\[
\begin{array}{llll}
\psi_1=(1,0,0),&
\psi_2=(0,1,0),&
\psi_3=(0,0,1),&
\psi_4=(1,\alpha,0),\\[1mm]
\psi_5=(1,-\alpha,0),&
\psi_6=(1,0,\beta),&
\psi_7=(1,0,-\beta),&
\psi_8=\left(0,\alpha,\beta\right),\\[1mm]
\psi_9=\left(0,\alpha,-\beta\right),&
\psi_{10}=(1,\alpha,\beta),&
\psi_{11}=(1,\alpha,-\beta),&
\psi_{12}=(1,-\alpha,\beta),\\[1mm]
\psi_{13}=(1,-\alpha,-\beta),&
\psi_{14}=(1,0,b i\beta),&
\psi_{15}=\left(0,\alpha,bi\beta\right),&
\psi_{16}=(1,-b i\alpha,0),\\[1mm]
\psi_{17}=(1,\alpha,-b i\beta),&
\psi_{18}=(1,b i\alpha,-b i\beta),&
\psi_{19}=(1,b i\alpha,\beta)
\end{array}
\]
In particular, there does not exist an orthogonal representation of $G_{19}$ in $\mathbb{R}^3$. 
\end{theorem}
\begin{proof}
Let $\{\psi_a\}_{a \in [19]}$ be an orthogonal representation of $G_{19}$.  Note that the first three vertices of $G_{19}$ form a clique. Since orthogonal representations are being classified up to unitary equivalence, we may assume:
\[
\begin{array}{llll}
\psi_1=(1,0,0),&
\psi_2=(0,1,0),&
\psi_3=(0,0,1)
\end{array}
\]

The rest of the proof is then described visually by figure \ref{fig:g19_2}. We start by parameterizing the vectors associated to vertices with a bold outline, which we call base vertices. Looking at vertex four, for example, the fact that it is adjacent to vertex three implies that $\psi_4$ is of the form $\alpha' \ket{1} + \alpha \ket{2}$. Note however that we know from lemma \ref{lem:notcollinear} that $\psi_4$ is not collinear with $\psi_1$ or $\psi_2$, so that $\alpha \alpha'\neq 0$ necessarily. Since we are disregarding scaling, we can assume that $\alpha' = 1$ without loss of generality. Applying this logic to all remaining base vertices, the corresponding vectors can be parameterized as follows, given  $\alpha,\beta,\gamma,\delta \in \mathbb{C} \backslash \{0\}$:
\[
\begin{array}{llll}
\psi_4=(1,\alpha,0), \;
\psi_6=(1,0,\beta), \;
\psi_8=(0,1,\gamma), \;
\psi_{14}=(1,0,\delta)
\end{array}
\]
For the next step of the proof, let $G'_{19}$ be the subgraph of $G_{19}$ corresponding to the undashed edges in figure \ref{fig:g19_2}. The point is that the vectors that were assigned to the base vertices extend to a unique orthogonal representation of $G'_{19}$, using the following algorithm: we find a non-base vertex $v$ with undetermined $\psi_v$ whose two neighbors $w_1,w_2$ along the edges of the same color as $v$ have already determined $\psi_{w_i}$. By lemma \ref{lem:notcollinear}, it can be assumed that the $\psi_{w_i}$ are not collinear (as it turns out, the ones produced by the algorithm are not collinear for any values of the parameters), and hence the condition that $\psi_v$ be orthogonal to them pins it down up to scale. An expression for it can be obtained by taking the cross product of the conjugates of the $\psi_{w_i}$. The base vertices have been chosen so that this algorithm ends up determining $\psi_v$ for every vertex $v$, with a possible sequence of vertices being:
\[9, 18, 16, 13, 5, 11, 17, 15, 7, 10, 19, 12\]
Given the vectors that were assigned to the base vertices, the corresponding vectors in the rest of the induced orthogonal representation of $G'_{19}$ (in which some algebraic simplifications have been made for convenience) are the following:
\[
\begin{array}{llll}
\psi_5=(\bar\alpha,-1,0),&
\psi_7=(\bar\beta,0,-1),&
\psi_9=(0,\bar\gamma,-1),&
\psi_{10}=(1,\alpha,\beta), \\[1mm]
\psi_{11}=(\bar\beta,\alpha\bar\beta,-1),&
\psi_{12}=(\bar\alpha,-1,\bar\alpha\beta),&
\psi_{13}=(\bar\alpha\bar\beta,-\bar\beta,-\bar\alpha),&
\psi_{15}=(0,1,\bar\alpha\delta),\\[1mm]
\psi_{16}=(\gamma,-\delta,0),&
\psi_{17}=(\bar\delta,\alpha\bar\delta,-1),&
\psi_{18}=(\bar\delta,\bar\gamma,-1),&
\psi_{19}=(1,-\alpha\beta\bar\delta,\beta)
\end{array}
\]
We then impose the necessary algebraic relations so that this orthogonal representation of $G'_{19}$ is in fact a valid orthogonal representation of $G_{19}$, i.e. so that the orthogonality relations corresponding to the dashed edges are satisfied. These are: 
\[E(G_{19}) \backslash E(G'_{19}) = \{(9,10),(9,13),(8,11), (8,12),(16,19)\}\] 
The corresponding equations are:
\begin{align}
\beta &= \gamma \alpha \label{eq:ga-beta}\\
 \alpha  &= \beta \bar{\gamma} \label{eq:bg-alpha}\\
\gamma  &= \bar{\alpha} \beta \label{eq:ab-gamma}\\
\alpha \gamma \bar{\beta} &= 1 \label{eq:agb-one}\\
\gamma &= - \delta^2 \bar{\alpha} \bar{\beta} \label{eq:gamma-delta}
\end{align}
Substituting \ref{eq:ab-gamma} into \ref{eq:ga-beta} and \ref{eq:bg-alpha}, we get that $\bar{\alpha} \alpha = \bar{\beta}\beta = 1$. From here, the first four equations can be seen to be equivalent to the conditions that $|\alpha| = |\beta| = 1$ and $\gamma = \beta \overline{\alpha}$. Finally, equation \ref{eq:gamma-delta} gives $\beta^2 = -\delta^2$, hence $\delta = i b\beta$ for $b \in \{1, -1\}$. Substituting these expressions for $\gamma$ and $\delta$ yields the first part of the theorem.

That $G_{19}$ does not have an orthogonal representation in $\mathbb{R}^3$ can be seen from the classification. Alternatively, if such a representation existed, one would exist with $\alpha, \beta, \gamma, \delta$ all real. The condition $\beta^2 = -\delta^2$ that was previously derived would force $\beta = \delta = 0$, which contradicts equation \ref{eq:agb-one}. 
\end{proof}

This lets us show the following lemma: 
\begin{lemma} \label{lem:Aiszero}
Let $b \in \{1,-1\}$ and $\alpha_1,\alpha_2,\beta_1,\beta_2$ be complex numbers with unit magnitude, and let $\{\psi_v\}$ and $\{\phi_v\}$ be the orthogonal representations of $G_{19}$ given by theorem \ref{thm:classificationreps} with parameters $(\alpha_1, \beta_1, b)$ and 
$(\alpha_2, \beta_2, b)$ respectively. If $A$ is a $3 \times 3$ matrix such that, for all $v$, 
\begin{align}
\braket{\phi_v | A | \psi_v}  &= 0 \label{eq:Aorthogonality}
\end{align}
Then $A = 0$.
\end{lemma}
\begin{proof}
Write
\[A = \begin{pmatrix} a_{11} & a_{12} & a_{13} \\ a_{21} & a_{22} & a_{23} \\ a_{31} & a_{32} & a_{33} \end{pmatrix}\]
Applying condition \ref{eq:Aorthogonality} with $v \in \{1,2,3\}$, we get that $a_{1,1} = a_{2,2} = a_{3,3} = 0$. After this, applying it with $v \in \{4,16\}$, we get the equations:
\[\alpha_1 a_{1,2} + \bar{\alpha}_2 a_{2,1} = 0 \]
\[-bi\alpha_1 a_{1,2} + bi\overline{\alpha}_2 a_{2,1} = 0 \]
Substituting the first equation into the second and using the fact that $b\alpha_1 \neq 0$ gives $a_{1,2}=0$, after which the second equation gives $a_{2,1}=0$. By the same token, applying \ref{eq:Aorthogonality} with $v \in \{6, 14\}$ yields $a_{1,3} = a_{3,1} = 0$, and applying it with $v \in \{8, 15\}$ yields $a_{2,3} = a_{3,2} = 0$. Hence $A=0$, as desired.
\end{proof}

We can now complete the proof of the main theorem of the subsection:
\begin{proof}[Proof of theorem \ref{thm:non3coloring}]
Take $n \geq 3$. We proceed by contradiction. Suppose that $G_{n}$ does have a rank-one quantum $n$-coloring, so that, identifying $S_\mathbb{C}^{n-1}$ with $\mathbb{CP}^{n-1}$ (as the two graphs are homomorphically equivalent), there exist $n$ orthogonal representations $\{\psi^{k}_v\}_{k \in [n], v \in V(G_{n})}$ of $G_{n}$ satisfying, for all $v$ and all $k \neq k'$,
\[\braket{\psi^{k}_v | \psi^{k'}_v} = 0\]
Let $T$ be the subspace of $\mathbb{C}^{n}$ spanned by the first three basis vectors. Note that the vertices $K = \{1,2,3,20,\dots,16+n\}$ form a $n$-clique of $G_{n}$ by construction.  For a given $k \in [n]$, let $U_k \in U(n)$ be the unitary whose columns consist of the (normalized) vectors $\psi^{k}_v$ with $v \in K$ in increasing order. By construction, given $v \in [19]$, we have that $U_k^\dagger\ket{\psi^{k}_v} \in T$, and $U_k^\dagger\ket{\psi^{k}_v} = \ket{v}$ for $v \in \{1,2,3\}$. Hence, the $\{U_k^\dagger\ket{\psi^{k}_v}\}_{v \in [19]}$ form an orthogonal representation of $G_{19}$ in $T\cong \mathbb{C}^3$ in normal form. Such representations were classified by theorem \ref{thm:classificationreps}: let $(\alpha_k, \beta_k, b_k)$ be the corresponding parameters. Since there are only two possible values for the $b_k$, setting $m= \lceil \frac{n}{2} \rceil$, by reordering, we can assume that $b_k = b_{k'}$ for all $k,k' \in [m]$. For $k \in [m]$, set $V_k = U_1^\dagger U_k$. Then, for all $k \neq k'$ and all $v \in [19]$, we have:
\[(\bra{\psi^{{k'}}_v}U_{k'})V_{k'}^\dagger V_k (U_k^\dagger\ket{\psi^{k}_v}) =\braket{\psi^{k'}_v |\psi^{k}_v} = 0\]
Lemma \ref{lem:Aiszero} then implies that the $V_{k'}^\dagger V_k$ are of the form:
\[
V_{k'}^\dagger V_k=
\begin{pmatrix}
0_{3\times 3} & *\\
W_{k',k} & *
\end{pmatrix} \]
where the $W_{k,k'}$ are $(n-3) \times 3$ isometries, since $V_{k'}^\dagger V_k$ is unitary. Since $V_1 = I$, setting $W_k = W_{1,k}$ for $k > 1$, this shows that $W_{k'}^\dagger W_k = 0$ whenever $k \neq k'$ and $k,k' > 1$, i.e. the column spaces of the $W_k$ are all pairwise orthogonal for $k > 1$. This gives $3(m-1)$ linearly independent vectors in a vector space of dimension $n-3$, which forces $3(m - 1) \leq n-3$, which is easily verified to be impossible for any $n \geq 3$. This finishes the proof. 
\end{proof}
One would like to show more strongly that $\chi_q(G_{n}) = n+1$ for all $n$. This seems difficult: even classifying elements of Hom($G_{n}, Gr(r, 3r)$) beyond $r=1$ becomes much more challenging, as there doesn't appear to be a simple proof anymore that subspaces assigned to different vertices must be linearly independent, in which case the propagation algorithm that was used to prove theorem \ref{thm:classificationreps} is no longer applicable. Even if this were to be true and a classification could be obtained, it is not at all clear that lemma \ref{lem:Aiszero} could be generalized in such a way that the proof technique of theorem \ref{thm:non3coloring} could be applicable, as we used a fact about the orthogonal representations of $G_{19}$ that was quite special. Most likely, the best that can be done is to prove that this holds for the first few values of $n$ using sum-of-square techniques, although that too appears to be quite hard because the semidefinite programs at play are very large. 

\section{Quantum colorings of real spheres} \label{sec:realspheres}
\subsection{The recasting as a linear-algebraic problem and consequences} \label{subsec:introitus}
This section handles the problem of determining for which values of $n$ it holds that $\chi_q(S^{n-1}) = n$. We first record the following algebraic simplification, which greatly simplifies the problem: 
\begin{proposition}
\label{prop:redrsp}
For $n \geq 3$, $r\geq 1$, it holds that $\chi_q^{(r)}(S^{n-1}) = n$ if and only if there exist unitaries $V_1, \dots,V_n$ on $\mathbb{C}^n \otimes \mathbb{C}^r$ such that, for all $a \neq b$, 
\[(V_a^\dagger V_b)^{\Gamma_1} = -V_a^\dagger V_b\]
\end{proposition} \begin{proof}
By definition, we have that $\chi_q^{(r)}(S^{n-1}) = n$ if and only if there exist $n$ pairwise adjacent homomorphisms $f_a: S^{n-1} \to Gr(r, rn)$. From the structure theorem \ref{thm:tianxu}, such homomorphisms are exactly these of the following form, for some unitaries $V_1,\dots,V_n$ acting on $\mathbb{C}^n \otimes \mathbb{C}^r$: 
\[f_a(\ket{\psi}\bra{\psi}) = V_a(\ket{\psi} \bra{\psi} \otimes I_r) V_a^\dagger\]
These homomorphisms are pairwise adjacent precisely when, for all $\psi \in S^{n-1}$ and $a,b \in [n]$ with $a \neq b$, 
\[(\bra{\psi} \otimes I_r) V_a^\dagger V_{b} (\ket{\psi} \otimes I_r) = 0\]
Lemma \ref{lem:polarization} shows that this holds precisely when the $V_a$ satisfy the condition in the statement of the proposition. 
\end{proof}

As a warm-up and to illustrate why Clifford algebras turn out to be essential for determining when unitaries as in the statement of proposition \ref{prop:redrsp} exist, we can handle the case of rank-one colorings without using the machinery of the next subsection. The following proves theorem \ref{thm:quaternions} and corollary \ref{cor:rankone}:
\begin{corollary}
    It holds that $\chi^{(1)}_q(S^{n-1}) = n$ if and only if $n \in \{2,4,8\}$.
\end{corollary}
\begin{proof}
    Since $S^1$ is a bipartite graph, the given condition does hold for $n=2$. Hence we can assume $n \geq 3$. Proposition \ref{prop:redrsp} then implies that the first part holds if and only if there exist matrices $V_1, \dots, V_n \in U(n)$ such that $V_a^\dagger V_b = -(V_a^\dagger V_b)^T$ for all $a \neq b$. 
    
    For the forward implication, note that we can assume without loss of generality that $V_1 = I$, by pre-multiplying the $V_a$ by $V_1^\dagger$. We then get that $V_a = -V_a^T$ for all $a > 1$, and therefore $V_a^\dagger V_b = -V_b^T (V_a^T)^\dagger = -V_b V_a^\dagger$, for all $a, b > 1$, which, upon left- and right-multiplying both sides by $V_a$, gives $V_aV_b = -V_bV_a$. 
    
    Let $\mathcal{A}$ be the $*$-algebra of $M_n(\mathbb{C})$ generated by $V_2,\dots,V_n$. Theorem \ref{thm:structure} then implies the existence of an orthogonal decomposition
    \[\mathbb{C}^n = \bigoplus_k T_k\]
    with the $T_k$ nonzero $\mathcal{A}$-invariant subspaces such that the restriction of $\mathcal{A}$ to any $T_k$ is irreducible. Fix such a value of $k$ and let $U_1, \dots,U_{n-1}$ be the restriction of the matrices $V_2,\dots,V_n$ to $T_k$.  Since $U_a U_b = -U_b U_a$ for all $a \neq b$, we see that $U_a^2$ commutes with all the $U_b$ and therefore their adjoints. Since the $U_b$ generate all of End$(T_k)$, it follows that the $U_a^2$ are all multiples of the identity: by absorbing the corresponding phase factor into the $U_a$, it can be assumed without loss of generality that $U_a^2 = I$ for all $a$. Hence, the $U_a$ form an irreducible representation of the real Clifford algebra on $n-1$ generators, which, by corollary \ref{cor:superclifford}, implies that the dimensions of the $T_k$ are all divisible by $d_{n-1}$. Hence, $n$ must be divisible by $d_{n-1}$. It is then simple to check that the only such values of $n$ are as given in the statement of the corollary.     

    For the converse, such unitaries can be exhibited directly in terms of the usual Pauli matrices. For $n=4$, set:
    \[V_1 = I, V_2 = Y \otimes I, V_3 = X \otimes Y, V_4 = Z \otimes Y\]
    For $n=8$, set:
    \[V_1 = I, V_2 = I \otimes I \otimes Y, V_3 = I \otimes Y \otimes Z, V_4 = X \otimes Y \otimes Z, V_5 = Y \otimes I \otimes Z, V_6 = Y \otimes X \otimes X, V_7 = Y \otimes Z \otimes X, V_8 = Z \otimes Y \otimes Z \]
    It can be checked that the $V_a$ above satisfy the required conditions.    
\end{proof}

\subsection{The correspondence with quantum error correction} \label{subsec:equivqecc}
This subsection derives algebraic conditions for the existence of a quantum $n$-coloring of $S^{n-1}$. The main result is the following (refer to proposition \ref{prop:codespace} for the definition of a maximal code space):
\begin{theorem} \label{thm:equiverrorcorrection}
For $n \geq 3$, $r \geq 1$, it holds that $\chi_q^{(r)}(S^{n-1}) = n$ if and only if there exists a unitary representation $\mathcal{U}$ of the Clifford algebra on $n-1$ generators in $U(nr)$ 
along with a maximal code space $P$ for $\mathcal{U}$.
\end{theorem}

We prove a few lemmata first. The following lemma provides a more workable angle of attack for the condition given by proposition \ref{prop:redrsp}: 
\begin{lemma} \label{lem:realstructure}
Let $n, r \geq 1$. Suppose that $V_1,\dots,V_n$ and $W_1, \dots, W_n$ are matrices acting on $\mathbb{C}^n \otimes \mathbb{C}^r$ which satisfy, for all $a,b \in [n]$ and $j \in [r]$, 
\[V_a \ket{b} \ket{j} = W_b \ket{a} \ket{j}\]

The following are then equivalent:
\begin{enumerate}
\item The $V_a$ are unitary, and for all $a,b \in [n]$ with $a \neq b$, 
\begin{align}(V_a^\dagger V_{b})^{\Gamma_1} = -V_a^\dagger V_{b}\label{eq:antisymmetric}\end{align}
\item The $W_a$ are unitary and satisfy:
\begin{align} W_a^\dagger W_{b}  + W_{b}^\dagger W_{a} = 2\delta_{a,b} I \label{eq:pseudoclifford}
\end{align} 
and, for all $a,b,c \in [n]$ with $a \neq b$, letting $P_c = \ket{c} \bra{c} \otimes I$, 
\begin{align} P_c W_a^\dagger W_{b} P_c = 0 \label{eq:killedbyP} \end{align}
\end{enumerate}
\end{lemma}

\begin{proof}
First, if the $V_a$ are unitary, we have that, for $a,b,c \in [n]$, $j,k \in [r]$,
\begin{align*}
\bra{c} \bra{j} W_a^\dagger W_{b} \ket{c} \ket{k} &= \bra{a} \bra{j} V_c^\dagger V_{c} \ket{b} \ket{k}\\
&= \delta_{a,b} \delta_{j,k}
\end{align*}
Thus, equation \ref{eq:killedbyP} is satisfied. Conversely, we see that if the $W_a$ are unitary and \ref{eq:killedbyP} holds, then
\[\bra{a} \bra{j} V_c^\dagger V_{c} \ket{b} \ket{k} = \delta_{a,b} \delta_{j,k}\]
Thus, the $V_c$ are unitary. Also, for $a, b,c,d\in [n]$ with $c \neq d$, we have:
\begin{align*}
\bra{c} \bra{j} W_a^\dagger W_{b} \ket{d} \ket{k} &= \bra{a} \bra{j} V_c^\dagger V_{d} \ket{b} \ket{k}\\
&= -\bra{b} \bra{j} V_c^\dagger V_{d} \ket{a} \ket{k}\\
&= -\bra{c} \bra{j} W_{b}^\dagger W_{a} \ket{d} \ket{k} 
\end{align*}
The $a=b$ special case of the above along with the previous calculation shows that if (1) holds, then the $W_a$ are unitary. Similarly, the $a \neq b$ case of the above along with the previous calculation shows that if (1) holds, then so does equation \ref{eq:pseudoclifford}, as desired. 

It remains to show that equation \ref{eq:antisymmetric} is implied by (2). For $a,b,c \in [n]$ with $a \neq b$ and $j,k \in [r]$, 
\[\bra{c} \bra{j} V_a^\dagger V_b \ket{c} \ket{k} = \bra{a} \bra{j} W_c^\dagger W_c \ket{b} \ket{k} = 0\]
Also, for $c,d \in [n]$ with $c \neq d$, 
\[\bra{c} \bra{j} V_a^\dagger V_b \ket{d} \ket{k} = -\bra{a} \bra{j} W_d^\dagger W_c \ket{b} \ket{k} = -\bra{d} \bra{j} V_a^\dagger V_b \ket{c} \ket{k} \]
which shows that equation \ref{eq:antisymmetric} holds, as desired.
\end{proof}

We will also need the following consequence of the classification of finite-dimensional $*$-algebras, here proven from scratch for completeness:
\begin{lemma} \label{lem:normalform}
Suppose that $\{M_{a,b}\}_{a,b\in [n]}$ are $D \times D$ matrices such that, for all $a,b,c,d$, 
\begin{enumerate}
\item $M_{a,b} M_{c,d} = \delta_{b,c} M_{a,d}$
\item $M_{a,b}^\dagger = M_{b,a}$
\item $\sum_a M_{a,a} = I$
\end{enumerate}
Then $D$ is divisible by $n$, and writing $D=nr$, there exists an isometry $U: \mathbb{C}^D \to \mathbb{C}^n \otimes \mathbb{C}^r$ such that, for all $a,b$, 
\[U M_{a,b} U^\dagger = \ket{a} \bra{b} \otimes I_r\]
\end{lemma}
\begin{proof}
Note first that hypotheses 1 and 2 imply that the $M_{a,a}$ are orthogonal projectors. Since $M_{a,a} = M_{a,1} M_{1,1} M_{a,1}^\dagger$ and $M_{1,1} = M_{1,a} M_{a,a} M_{1,a}^\dagger$, it follows from the submultiplicativity of rank that the $M_{a,a}$ all have the same rank as $M_{1,1}$, call it $r$. Hypothesis 1 also gives $M_{a,a} M_{b,b} = 0$ whenever $a \neq b$, so that the $M_{a,a}$ are pairwise orthogonal projectors. Taking traces in hypothesis 3 then implies that $D=nr$. As $j$ runs over $[r]$, let the $\ket{\psi_{1,j}}$ be an orthonormal basis of the range of $M_{1,1}$, and, for $a \in [n]$, set $\ket{\psi_{a,j}} = M_{a,1} \ket{\psi_{1,j}}$. For any $a,b,c,d \in [n], j,k \in [r]$, we see that:
\begin{align*}
\braket{\psi_{a,j} | M_{c,d}  | \psi_{b,k}} &= \braket{\psi_{1,j} | M_{1,a} M_{c,d} M_{b,1} | \psi_{1,k}}\\
                                    &= \delta_{a,c} \delta_{b,d} \delta_{j,k}
\end{align*}
Note first that setting $a=b=c=d$ in the above shows that, for any $a$, the $\ket{\psi_{a,j}}$ form an orthonormal set inside the range of $M_{a,a}$, and hence constitute an orthonormal basis for said range. Since the ranges of differing $M_{a,a}$ are orthogonal, this implies that the $\ket{\psi_{a,j}}$ form an orthonormal basis of $\mathbb{C}^D$. Hence, the linear map $U:  \mathbb{C}^D \to \mathbb{C}^n \otimes \mathbb{C}^r$ specified by
\[U\ket{\psi_{a,j}} = \ket{a} \ket{j}\]
is an isometry. It is simple to see from the previous calculation that $U$ satisfies the conditions of the lemma. 
\end{proof}

With these lemmata in hand, we can show:

\begin{proof}[Proof of theorem \ref{thm:equiverrorcorrection}]
From proposition \ref{prop:redrsp} and lemma \ref{lem:realstructure}, we get that $\chi^{(r)}_q(S^{n-1}) = n$ if and only if there exist unitaries $W_1, \dots, W_n$ on $\mathbb{C}^n \otimes \mathbb{C}^r$ with $W_a^\dagger W_{b} + W_{b}^\dagger W_{a} = 2\delta_{a,b} I$ for all $a,b$ and with $P_cW_a^\dagger W_{b} P_c = 0$ for all $a \neq b$ and for all projectors $P_c$ of the form $P_c = \ket{c}\bra{c} \otimes I$ for $c \in [n]$. We prove that this is equivalent to the second part of the theorem.

We start by proving the direct implication, where we identify $\mathbb{C}^{nr}$ with $\mathbb{C}^n \otimes \mathbb{C}^r$ in the natural way. Suppose that such unitaries $W_a$ do exist, and for $a \in [n-1]$, set $U_a = i W_1^\dagger W_{a+1}$. We see that 
\begin{align*}
U_a^\dagger &= -i W_{a+1}^\dagger  W_1 \\
            &= -i (-W_1^\dagger W_{a+1})\\
            &= U_a
\end{align*}
Also, for $a \neq b$, 
\begin{align*}
U_{a} U_b &= U_{b}^\dagger U_a \\
         &= W_{a+1}^\dagger W_{b+1}\\
         &= -W_{b+1}^\dagger W_{a+1}\\
         &= - U_b U_a
\end{align*}
Hence, the $U_a$ satisfy the Clifford relations. The projector $P_1 = \ket{1} \bra{1} \otimes I_r$ has rank $r$ and satisfies $P_1 U_aP_1 = P_1 U_a^\dagger U_b P_1 = 0$ for all $a \neq b$, by assumption, and is therefore a maximal code space for the $U_a$, as per the first statement in proposition \ref{prop:codespace}.

We now show the converse. Let the $U_a$ and $P$ as in the statement of the theorem be given. Set $\tilde{W}_1 = I$, and for $a\in [n] \backslash \{1\}$, set $\tilde{W}_a = i U_{a-1}$. It is simple to verify that it holds that $\tilde{W}_a^\dagger \tilde{W}_b + \tilde{W}^\dagger_b \tilde{W}_a = 2\delta_{a,b}I$ for all $a,b$. Also, from proposition \ref{prop:codespace}, we have that $P \tilde{W}_a^\dagger \tilde{W}_bP = 0$ for all $a \neq b$. 

Set $M_{a,b} = \tilde{W}_aP\tilde{W_b}^\dagger$. First, we see:
\begin{align*}
M_{a,b} M_{c,d} &= \tilde{W}_a (P \tilde{W_b}^\dagger \tilde{W_c} P) \tilde{W}_d^\dagger\\
               &= \delta_{b,c} \tilde{W}_a P^2\tilde{W}_d^\dagger\\
               &= \delta_{b,c} M_{a,d}
\end{align*}
where the hypothesis that $PU_aP = PU_a U_bP = 0$ for all $a \neq b$ was used in the second line. We also have that $M_{a,b}^\dagger = M_{b,a}$ since $P$ is Hermitian. Finally, the second statement in proposition \ref{prop:codespace} shows that the $M_{a,a}$ sum to the identity. Hence, by lemma \ref{lem:normalform}, there exists an isometry $U: \mathbb{C}^{nr} \to \mathbb{C}^n \otimes \mathbb{C}^r$ such that, for all $a,b \in [n]$, 
\[UM_{a,b}U^\dagger = \ket{a} \bra{b} \otimes I_r\]
For every $a \in [n]$, set $W_a = U\tilde{W}_aU^\dagger$. We see that we still have that $W_a^\dagger W_b + W_b^\dagger W_a = 2 \delta_{a,b} I$ for all $a,b$. Also, take $j \in [n]$ and set $P_j = \ket{j} \bra{j} \otimes I$. We have, for $a \neq b$, 
\begin{align*}
P_j W_a^\dagger W_b P_j &= U(M_{j,j} \tilde{W}_a^\dagger \tilde{W}_b M_{j,j})U^\dagger \\
                        &= U(\tilde{W}_j P \tilde{W}_j^\dagger \tilde{W}_a^\dagger \tilde{W}_b  \tilde{W}_j P \tilde{W}_j^\dagger)U^\dagger\\
                        &= \pm U\tilde{W}_j (P \tilde{W}_a^\dagger \tilde{W}_b  \tilde{W}_j^2 P) \tilde{W}_j^\dagger U^\dagger\\
                        &=0
\end{align*}
where, in the second to last line, we used the fact that the $\tilde{W}_a$ are all either the identity or scaled Clifford algebra generators, and therefore either commute or anticommute and square to the identity. Hence, the constructed $W_a$ have the required property. 
\end{proof}

\subsection{Proving the main theorem} \label{subsec:theproof}
With the results of the previous subsection in hand, we can complete the proof of theorem \ref{thm:realspheres}. Note that point 4 of theorem \ref{thm:realspheres} already follows from theorem \ref{thm:equiverrorcorrection} combined with corollary \ref{cor:superclifford} (as the existence of the $U_i$ forces $nr$ to be divisible by $d_{n-1}$). 

We start by determining the spectrum of the noise channel corresponding to an irreducible representation of a Clifford algebra. The reader will find the relevant definitions in subsections \ref{subsec:qm} and \ref{subsec:reptheory}.
\begin{lemma} \label{lem:eigvalsclifford}
    For a given $m \geq 1$, let $\mathcal{U} = \{\Gamma_1, \dots, \Gamma_m\}$ be an irreducible unitary representation of the Clifford algebra on $m$ generators in $U(d_m)$. The constant degree spaces $\Lambda_d$ are all eigenspaces of $\Phi_\mathcal{U}$ viewed as a linear operator on $M_{d_m}(\mathbb{C})$, with corresponding eigenvalues:
    \[\lambda_d = (m-2d)(-1)^d\]
    Since the representation is irreducible, this is the full spectrum of $\Phi_\mathcal{U}$.
\end{lemma}
\begin{proof}
Given $S \subseteq [m]$  and $a \in [m]$, applying lemma \ref{lem:cliffordmonomials}, we see that $\Gamma_a \Gamma_S \Gamma_a = (-1)^{|S| - I[a \in S]} \Gamma_S$. It is then immediate from the definitions that
\[\Phi_\mathcal{U}(\Gamma_S) = \lambda\Gamma_S\]
with
\[\lambda = |S|(-1)^{|S|-1} + (m-|S|)(-1)^{|S|}\]
Rearranging then gives the lemma. 
\end{proof}

We start by proving the following, which shows part one of theorem \ref{thm:realspheres}: 
\begin{theorem} \label{thm:divisibleby4}
Suppose that, for $n \geq 3$, $\chi_q(S^{n-1}) = n$. Then $n$ is a multiple of 4.
\end{theorem}
\begin{proof}
Writing $m=n-1$, for some rank $r \geq 1$, let $\mathcal{U} = \{U_1,\dots,U_m\} \subseteq U(nr)$ and $P$ be the data that is promised by theorem \ref{thm:equiverrorcorrection}. Note that we may conjugate the $U_i$ and $P$ by any given unitary and preserve the setup. Hence, from corollary \ref{cor:superclifford}, we can assume without loss of generality that, for some natural numbers $p, q \geq 0$ with $nr = d_m(p+q)$, we have that
\begin{equation} U_a = (I_p \otimes \Gamma_a) \oplus (I_q \otimes (-\Gamma_a)) \label{eq:blockformUi} \end{equation}
for some irreducible representation $\Gamma_a$ of the Clifford algebra on $m$ generators in $U(d_m)$. 

Writing $X = P - I/n$, it follows from point three of proposition \ref{prop:codespace} that $X$ is an eigenvector of $\Phi_{\mathcal{U}}$ with eigenvalue $-1$. 
 Writing $\Gamma = \{\Gamma_1,\dots,\Gamma_m\}$, we see that if we write:
\[X = \sum_{a, b \in [p+q]} \ket{a} \bra{b} \otimes X_{a,b}\]
Then, from the block decomposition \ref{eq:blockformUi}, we get:
\[X = -\Phi_\mathcal{U}(X) = \sum_{a, b \in [p+q]} (-1)^{I[a > p] + I[b > p] + 1} \ket{a} \bra{b} \otimes \Phi_\Gamma(X_{a,b})\]
Hence, for all $a, b$, we get that 
\[\Phi_\Gamma(X_{a,b}) = (-1)^{I[a > p] + I[b > p] + 1} X_{a,b}\]
From lemma \ref{lem:eigvalsclifford}, we see that for even $m$, the eigenvalues of $\Phi_\Gamma$ are all multiples of two, so that neither 1 nor $-1$ is an eigenvalue. The above would then force $X = 0$, hence $P = I/n$, which is not a projector unless $n=1$.  Hence, $n$ must be even. 

We now show that $m \equiv 1 \mod 4$ is also impossible, which finishes the proof of the theorem. In this case, lemma \ref{lem:eigvalsclifford} can be seen to imply that $-1$ is not in the spectrum of $\Phi_\Gamma$, but $1$ is. This forces $X_{a,b} = 0$ whenever $a,b \leq p$ or $a,b > p$. Hence, $X$ has block form:
\[X = \begin{pmatrix} 0 & B \\ C & 0 \end{pmatrix}\]
Therefore:
\[P = X+I/n = \begin{pmatrix} I_{d_mp}/n & B \\ C & I_{d_mq}/n \end{pmatrix}\]
We calculate:
\[P^2 = \begin{pmatrix} I_{d_mp}/n ^2 + BC & 2B/n \\ 2C/n & I_{d_mq}/n ^2 + CB
\end{pmatrix}\]
Since $P$ is assumed to be a projector, we have $P^2 = P$. Comparing the off-diagonal blocks yields $2B/n = B$ and $2C/n = C$. Since we are assuming that $n \geq 3$, this forces $B,C=0$, which implies $P = I/n$, which is again impossible. 
\end{proof}

We now turn to proving the second part of theorem \ref{thm:realspheres}. First, we show the following simple lemma. The converse is easily seen to hold, but is not stated as it is not needed for our purposes.

\begin{lemma} \label{lem:existenceofsubsets}
    Suppose that $n \geq 2$ is such that there exists a Hadamard matrix of order $n$. Then, there exist subsets $S_1,\dots,S_n$ of $[n-1]$ such that, for all $a \neq b$, 
        \[|S_a \triangle S_b| =\frac{n}{2} \]
\end{lemma}
\begin{proof}
Let $M$ be such a Hadamard matrix. By multiplying columns by $-1$ appropriately (which preserves the property of being a Hadamard matrix), we can arrange for $M$ to be such that $M_{n,a} = 1$ for all $a$. For a given $a \in [n]$, set
\[S_a = \{b \in [n-1] \mid M_{b,a} = 1\}\]
Then, using the normalization of $M$,
\[|S_a \triangle S_b| = |\{c\in [n] \mid M_{c,a} \neq M_{c,b}\}|\]
Since, by hypothesis, the columns $a$ and $b$ of $M$ agree on exactly half of the rows, we get that, for $a \neq b$, 
\[|S_a \triangle S_b| = \frac{n}{2}\] 
As desired. 
\end{proof}

We then show the following, which gives part two of theorem \ref{thm:realspheres}: 
\begin{theorem}
If a given multiple of four $n$ is such that a Hadamard matrix of order $n$ exists, we have
\[\chi_q^{(d_{n-1})}(S^{n-1}) = n\] 
\end{theorem}
\begin{proof}
We will construct the data requested by theorem \ref{thm:equiverrorcorrection} to establish the result. Set $m=n-1$ and let $\Gamma_1, \dots,\Gamma_m$ be an irreducible unitary representation of the Clifford algebra in $U(d_m)$. For $a \in [m]$, set:
\[U_a = \Gamma_a \otimes I_n\]
We will prove the existence of a maximal code space $P$ for the $U_a$. Since we are assuming that a Hadamard matrix of order $n$ exists, lemma \ref{lem:existenceofsubsets} implies the existence of subsets $S_1, \dots, S_n$ of $[m]$ with $|S_a \Delta S_b| = n/2$ for all $a \neq b$. Set 
\[V = \frac{1}{\sqrt{n}} \sum_{a=1}^n \Gamma_{S_a} \otimes \ket{a}\]
\[P = VV^\dagger = \frac{1}{n} \sum_{a,b=1}^n \Gamma_{S_a}^\dagger \Gamma_{S_b} \otimes \ket{a} \bra{b} \]
First note that $V^\dagger V = I_{d_m}$, so that $P$ is an orthogonal projector. To show that $P$ is a maximal code space for the $U_a$, setting $X = P-I/n$, by part three of proposition \ref{prop:codespace}, it is enough to show that $X$ lies in the $(-1)$-eigenspace of $\Phi_\mathcal{U}$. Note that:
\[X =\frac{1}{n} \sum_{\substack{a,b=1 \\ a \neq b}}^n \Gamma_{S_a}^\dagger \Gamma_{S_b} \otimes \ket{a} \bra{b} \]
It follows from lemmata \ref{lem:cliffordmonomials} and \ref{lem:productmonomials} that $\Gamma_{S_a}^\dagger \Gamma_{S_b} \in \Lambda_{|S_a \Delta S_b|} = \Lambda_{n/2}$ whenever $a \neq b$. Since any matrix in $\Lambda_{n/2}$ is an eigenvector of $\Phi_\Gamma$ with eigenvalue $-1$, the result follows.
\end{proof}

We now prove part three of theorem \ref{thm:realspheres}, which finishes the proof of the theorem. It should be noted that the fact that $S^{n-1}$ is quantumly $n$-colorable whenever $n$ is a power of two already follows from the previous theorem, as Sylvester's construction already implies the existence of Hadamard matrices of these orders (these are the familiar Hadamard towers from quantum computing, disregarding the normalization). We provide an additional construction which is entanglement-optimal, unlike the previous one. Note that the statement that $S^7$ admits a rank-one quantum $8$-coloring is recovered as a special case.
\begin{theorem} \label{thm:colorspoweroftwo}
    For all powers of two $n \geq 8$, setting $r = \frac{d_{n-1}}{n}$, it holds that:
    \[\chi_q^{(r)}(S^{n-1}) = n\] 
\end{theorem}
\begin{proof}
Write $n = 2^s$ and set $m=n-1$. We again construct the data required by theorem \ref{thm:equiverrorcorrection}. Let $\{\Gamma_v\}_{v \in \mathbb{F}^s_2 \backslash \{0\}}$ be an irreducible representation of the Clifford algebra on $m$ generators in $U(d_m)$. Given $w \in \mathbb{F}^s_2 \backslash \{0\}$, define:
\[L_w = \{v \in \mathbb{F}^s_2 \mid v \cdot w = 1\}\]
It is a standard fact that $|L_w| = 2^{s-1}$ for all $w$. Also, given $w, w' \in \mathbb{F}^s_2 \backslash \{0\}$ with $w \neq w'$, we have that $L_w \Delta L_{w'} = L_{w + w'}$, and so $|L_w \cap L_{w'}| = 2^{s-2}$. Since it is assumed that $s \geq 3$, lemma \ref{lem:cliffordmonomials} implies that $\Gamma_{L_w}$ and $\Gamma_{L_{w'}}$ commute and that $\Gamma_{L_w}^\dagger = \Gamma_{L_w}$, and hence $\Gamma_{L_w}^2 = I$. 

Let $e_1, \dots, e_s$ be the standard basis of $\mathbb{F}^s_2$, and set $P_a = \frac{1}{2}(I+\Gamma_{L_{e_a}})$. It is easy to check that the $P_a$ are orthogonal projectors and commute, and therefore that, setting
\[P = \prod_{a=1}^s P_a\]
we have that $P$ is an orthogonal projector. We show that $P$ is a maximal code space for the $\Gamma_v$. Setting $X = P - I/n$, we get
\begin{align*}
X  &= \frac{1}{n} \sum_{w \in \{0,1\}^s \backslash \{0\}}\prod_{a=1}^s \Gamma_{L_{e_a}}^{w_a} \\
&= \frac{1}{n} \sum_{w \in \mathbb{F}_2^s \backslash \{0\}} b_w \Gamma_{L_w}\\
\end{align*}
for some $b_w \in \{1,-1\}$, where lemma \ref{lem:cliffordmonomials} was used to simplify the product. Lemma \ref{lem:eigvalsclifford} shows that $X$ lies in the $(-1)$-eigenspace of $\Phi_\Gamma$, since $|L_w| = n/2$. This completes the proof.  
\end{proof}

\section{Concluding remarks and open problems}
Our results investigated possible directions in which the construction \ref{thm:quaternions} could be extended. We showed that no such extension is possible if the hypothesis that the orthogonal representation be real is lifted, while an extension is indeed possible whenever $n$ is divisible by four provided that a Hadamard matrix of order $n$ exists. We have also shown that the rank of the corresponding colorings must grow exponentially with $n$.

Our results leave open the following avenues for future research:
\begin{enumerate}
    \item Our techniques only apply to the study of the finite-dimensional chromatic numbers $\chi_q(S^{n-1})$. In particular, they are not strong enough to show nontrivial lower bounds on the commuting chromatic numbers $\chi_{qc}(S_\mathbb{F}^{n-1})$, and it would be interesting to look for alternate ways to lower bound them. Conjecture: $\chi_{qc}(S_\mathbb{F}^{n-1}) = \chi_{q}(S_\mathbb{F}^{n-1})$ always. 
    \item Corollary \ref{cor:asymptotics} showed that $\chi_q(S^{n-1}) = n + o(n)$. However, our results have no bearing on the asymptotic behavior of $\chi^{(r)}_q(S^{n-1})$ when $r$ is either a constant or is a polynomial function of $n$, and it could be that this grows exponentially fast with $n$, just like the ordinary chromatic number. Determining whether this is the case appears to be an interesting research project.
    \item We have shown that for all $n \geq 3$ which are multiples of four, the sphere $S^{n-1}$ is quantumly $n$-colorable if a Hadamard matrix of order $n$ exists. It would be interesting to either prove this unconditionally or to show that the reverse implication holds.  In this regard, theorem \ref{thm:colorspoweroftwo} serves as a warning in that the construction we gave can be improved upon when $n$ is a power of two, though it could be that our construction is rigid for other values of $n$. 
    \item We have determined for which values of $n$ there exists an exact remote state preparation scheme for $S^{n-1}$ with communicated register size $n$, subject to the condition that this scheme use a maximally entangled state and projective measurements. Does the same classification hold without this condition?
    \item Another interesting avenue is that of effectiveness. As per the generalized de Bruijn-Erd\H{o}s theorem \ref{thm:debruijnerdoswatereddown}, most theorems we have shown have analogues for finite graphs, though the exact analogues one would hope for would require proving the existence of a subgraph $G$ of $S_\mathbb{F}^{n-1}$ with $\chi_q(G') = \chi_q(S_\mathbb{F}^{n-1})$, which, as was discussed, is not generically true for infinite graphs. It would be interesting to give explicit examples of such subgraphs.
\end{enumerate}

\section{Acknowledgments}
We thank Debbie Leung for having pointed us to the references \cite{Zeng_2002, Leung_2003} as well as William Slofstra, Ashwin Nayak and Mike Brannan for many useful discussions. We also thank Richard Cleve for his help with proofreading parts of the paper. Although everything in the paper was exclusively written by the author, we acknowledge the use of AI tools such as Claude or GPT-5 for researching the literature and for dealing with some of the technical parts of the paper. We acknowledge the support of the Natural Sciences and Engineering Research Council of Canada (NSERC) grants ALLRP-578455-2022 and RGPIN-2023-03731.

\printbibliography

\end{document}